\documentclass[11pt,a4paper,fleqn]{article}

\usepackage[utf8]{luainputenc}

\usepackage{fullpage}

\usepackage{tgtermes}

\usepackage{amssymb,
  amsmath,
  amsfonts,
  amsthm}

\usepackage{enumitem}
\setlist{nolistsep,leftmargin=\parindent}
\setitemize{label={--}}
\setenumerate{font=\footnotesize}

\theoremstyle{plain}
\newtheorem{theorem}{Theorem}[section]
\newtheorem{lemma}[theorem]{Lemma}
\newtheorem*{claim*}{Claim}

\newtheorem{fact}[theorem]{Fact}
\newtheorem*{fact*}{Fact}
\newtheorem{proposition}[theorem]{Proposition}

\theoremstyle{definition}

\theoremstyle{remark}
\newtheorem{remark}[theorem]{Remark}

\newcommand\Nat{\mathbb{N}}
\newcommand\Int{\mathbb{Z}}

\newcommand\col{\operatorname{col}}
\newcommand\can{\operatorname{can}}
\newcommand\inv{\operatorname{inv}}

\newcommand\tw{\operatorname{tw}}
\newcommand\impr{\operatorname{impr}}

\newcommand\dtm{\textsc{dtm}}

\renewcommand{\phi}{\varphi}

\newcommand{\seq}{\textnormal{seq}}
\newcommand{\dec}{\textnormal{dec}}

\newcommand{\wo}{\prec}
\newcommand{\incomp}{\equiv}

\newcommand{\seqwo}{\wo_{\seq}}
\newcommand{\seqincomp}{\incomp_{\seq}}

\newcommand{\decwo}{\wo_{\dec}}
\newcommand{\decincomp}{\incomp_{\dec}}

\newcommand{\seqdecwo}{\wo_{(\dec,\seq)}}
\newcommand{\seqdecincomp}{\incomp_{(\dec,\seq)}}

\newcommand{\proccomp}{\textsc{comparison}}

\newcommand{\stack}{\textnormal{stack}}

\newcommand\Class[1]{\mathchoice{\text{\normalfont\small$\mathrm{#1}$}}{\text{\normalfont\small$\mathrm{#1}$}}{\text{\normalfont$\mathrm{#1}$}}{\text{\normalfont$\mathrm{#1}$}}} 
\newcommand{\Lang}[1]{\ifmmode{\text{\textsc{#1}}}\else\textsc{#1}\fi}

\usepackage{color}
\definecolor{darkred}{rgb}{0.5,0,0}
\definecolor{darkblue}{rgb}{0,0,0.5}

\usepackage{tikz}
\usetikzlibrary{graphs,
  trees,
}

\usepackage{url} 
\usepackage[pdfborder={0 0 0}]{hyperref}

\begin{document}

\title{Canonizing~Graphs~of Bounded~Tree~Width in~Logspace}

\author{Michael Elberfeld \and Pascal Schweitzer}

\date{{\small RWTH} Aachen University\\
  Aachen, Germany\\
  {\small \texttt{\{elberfeld|schweitzer\}@informatik.rwth-aachen.de}}\\[2ex]
  \today}

\maketitle

\begin{abstract} 
  Graph canonization is the problem of computing a unique representative, a
  canon, from the isomorphism class of a given graph. This implies that two graphs
  are isomorphic exactly if their canons are equal. We show that graphs of bounded
  tree width can be canonized by logarithmic-space (logspace) algorithms. This
  implies that the isomorphism problem for graphs of bounded tree width can be
  decided in logspace. In the light of isomorphism for trees being hard for
  the complexity class logspace, this makes the ubiquitous class of graphs of
  bounded tree width one of the few classes of graphs for which the complexity of
  the isomorphism problem has been exactly determined.

\end{abstract}

\section{Introduction}
\label{sec:introduction}

The \emph{graph isomorphism problem} (\Lang{isomorphism}) -- deciding whether
two given graphs are the same up to renaming vertices -- is one of the few
fundamental problems in~$\Class{NP}$ for which we neither know that it is
polynomial-time solvable nor that it is $\Class{NP}$-complete. Since
$\Class{NP}$-hardness would imply a collapse of the polynomial hierarchy to its
second level~\cite{BoppanaHZ1987,Schoening1988}, significant effort has been put
into better understanding the graph-theoretic requirements on input graphs that
make \Lang{isomorphism} polynomial-time solvable. A classical result of
Bodlaender~\cite{Bodlaender1990} shows that \Lang{isomorphism} is solvable in
polynomial time for graphs of \emph{bounded tree
width}~\cite{Bodlaender1990}. Polynomial-time algorithms are also known for
other graph classes like planar graphs~\cite{HopcroftW1974,Tarjan1971} and more
general graphs with a crossing-free embedding into a fixed
surface~\cite{FilottiM1980,Grohe2000,Miller1980}. A deeper complexity-theoretic
insight behind the polynomial-time algorithms for embeddable graphs is given by
the fact that \Lang{isomorphism} for graphs embeddable into the
plane~\cite{Dattaetal2009} or a fixed surface~\cite{ElberfeldK2014} can be
decided by logarithmic-space (\emph{logspace}) algorithms. These algorithms,
which are polynomial-time algorithms using at most a logarithmic amount of
memory, define the complexity class $\Class{L}$.

So far, it has been an open problem whether for graphs of bounded tree width the
isomorphism problem can be solved in logspace. Guided by the goal to determine
the exact complexity of the isomorphism problem for these graphs, there has been
a sequence of ever stronger partial results. Bodlaender's
algorithm~\cite{Bodlaender1990} placing \Lang{isomorphism} for graphs of bounded
tree width in $\Class{P}$ was first refined to an upper bound in terms of
logarithmic-depth circuits with threshold gates (i.e., circuits defining the
complexity class $\Class{TC}^1$) \cite{GroheV2006} and later improved to use
semi-unbounded fan-in Boolean gates (i.e., circuits defining the complexity
class $\Class{SAC}^1$) \cite{DasTW2012}. Since the chain $\Class{L} \subseteq
\Class{SAC}^1 \subseteq \Class{TC}^1 \subseteq \Class{P}$ is all we know about
the inclusion relations of these classes, these works leave the question for a
logspace approach that applies to every class of graphs of bounded tree width
open. Logspace algorithms are known for small constant bounds on the tree
width. Indeed, Lindell's~\cite{Lindell1992} classical approach to testing
isomorphism of trees provides us with a logspace algorithm for graphs of
tree-width at most~1. This was generalized to graphs of tree width at most 2
\cite{Arvindetal2008} and results of \cite{DattaNTW2009} for graphs without
$K_5$ as a minor apply to graphs of tree width at most~3. Moreover, $k$-trees,
the maximal tree-width-$k$ graphs, admit logspace isomorphism
tests~\cite{Arvindetal2012} as well as graphs with a bounded tree
depth~\cite{DasER2015}. While providing us with ever larger classes of graphs
with logspace algorithms for \Lang{isomorphism}, the general question for
bounded tree width graphs remained open.

\paragraph{Results.} Our first main result answers the above question in
its most general way by showing that the isomorphism problem for graphs of
bounded tree width can be solved by logspace algorithms. Together with a result
of Jenner et al.~\cite{Jenneretal2003}, showing that the isomorphism problem for
trees is $\Class{L}$-hard, this pinpoints the complexity of \Lang{isomorphism}
for graphs of bounded tree width.

\begin{theorem}
  \label{th:isomorphism-tw} 
  For every positive $k \in \Nat$, there is a logspace algorithm that decides
  whether two given graphs of tree width at most $k$ are isomorphic. Moreover,
  this problem, is complete for~$\Class{L}$ with respect to first-order
  reductions.
\end{theorem}

For testing whether two graphs are isomorphic, it is in practice often helpful
to perform a two-step approach that first computes a canonical representative
for each isomorphism class, called the \emph{canon}, and then declares the two
graphs to be isomorphic exactly if their canons are equal (rather than
isomorphic).  To also be able to construct an \emph{isomorphism} between the
input graphs (that means, a bijective function between the vertex sets of given
graphs that preserves their edge relations), it is helpful to have additionally
access to an isomorphism from the input graphs to their canons. Such an
isomorphism to the canon is called a \emph{canonical labeling} of a graph. An
isomorphism between the input graphs can be constructed by composing canonical
labelings.

For most isomorphism algorithms that have been developed so far, it was
possible, with varying amounts of extra effort, to turn them into an algorithm
that computes canons and canonical labelings. Hence, deciding \Lang{isomorphism}
and computing canons are often known to have the same complexity. However, the
current situation for graphs of bounded tree width is different: While the
approach from~\cite{DasTW2012} puts the isomorphism problem for graphs of
bounded tree width into $\Class{SAC}^1$, this is not done by providing a
canonization procedure. In fact, the best known upper bound for canonizing
graphs of bounded tree width uses logarithmic-depth circuits with unbounded
fan-in Boolean gates (i.e., circuits defining the complexity class
$\Class{AC}^1$) \cite{Wagner2011}. Between these classes only the relation
$\Class{SAC}^1 \subseteq \Class{AC}^1$ is known. Our second main result
clarifies this situation by providing a logspace algorithm for canonizing graphs
of bounded tree width.

\begin{theorem}
  \label{th:canonization-tw} 
  For every $k \in \Nat$, there is a logspace algorithm that, on input of a
  graph $G$ of tree width at most~$k$, outputs an isomorphism-invariant encoding
  of $G$ (a canon) and an isomorphism to it (a canonical labeling).
\end{theorem}

\paragraph{Techniques.} 

The known logspace approaches for canonizing certain classes of bounded tree
width graphs are based on first computing an isomorphism-invariant tree
decomposition for the given input graph and then adjusting Lindell's tree
canonization approach to canonize the graph with respect to the
decomposition. For example, for $k$-trees~\cite{Arvindetal2012} an
isomorphism-invariant tree decomposition arises by taking a graph's maximal
cliques and their intersections as the bags of the decomposition and connecting
two bags based on inclusion. The resulting tree decomposition is both
isomorphism-invariant, which is required for a canonization procedure to be
correct, and has width $k$, which enables the application of an extension of
Lindell's approach by taking (the constant number of) orderings of the vertices
of the bags into account.

\emph{Technique 1: Isomorphism-invariant tree decomposition into bags without
clique separators.}  In general, for graphs of tree width at most~$k$, there is
no isomorphism-invariant tree decomposition of width~$k$. A simple example of
graphs demonstrating this are cycles, which have tree width 2, but no isomorphism-invariant
tree decomposition of width 2. We could hope to find an isomorphism-invariant
tree decomposition by allowing approximate tree decompositions (that means,
allowing an increase of the width to some constant~$k'$). Again, cycles show
that such tree decompositions do not always exist. To address this issue, we
could consider not just one tree decomposition, but an isomorphism-invariant and
polynomial-size collection of tree decompositions. However, for all~$k' \in
\Nat$, there are graphs of tree width at most~$3$ for which the smallest
isomorphism-invariant collection of tree decompositions of width~$k'$ has
exponential size. Simple graphs demonstrating this fact are given by forming the
disjoint union of~$n$ cycles of length~$n$, and adding a vertex that is adjacent
to every other vertex.

We work around this problem by considering isomorphism-invariant tree
decompositions that may have bags of unbounded size, but with bags that are
easier from a graph-theoretic and algorithmic perspective than the original
graph. An algorithm developed recently~\cite{Lokshtanovetal2014a} (which refined
the time complexity for \Lang{isomorphism} on graphs of tree width~$k$ from
Bodlaender's $n^{O(k)}$ bound to $g(k) \cdot n^{O(1)}$ for a function $g$)
applies a technique from Leimer~\cite{Leimer1993} that turns the input graph
into its isomorphism-invariant collection of maximal induced subgraphs without
clique separators called \emph{maximal atoms}. (In the example above, the
maximal atoms are exactly the enriched cycles.) We adapt this idea as a first
step for the proofs of our main results, but need to adjust it both with respect
to the graph-theoretic concepts as well as the algorithmic ideas involved: While
the collection of subgraphs produced by Leimer's approach is
isomorphism-invariant, the tree underlying the resulting decomposition highly
depends on the order in which subgraphs are considered and, thus, is not
isomorphism-invariant. While it is always sufficient to have an
isomorphism-invariant set of potential bags capturing a tree decomposition in
order to perform polynomial-time isomorphism tests (see~\cite{OtachiS2014}), in
order to apply or work towards logspace techniques it is necessary to have an
isomorphism-invariant tree decomposition. Our first main technical contribution
develops a tree decomposition of a graph whose bags are maximal atoms that is
isomorphism-invariant, and logspace-computable for graphs of bounded tree width.

\emph{Technique 2: Nested tree decomposition and a quasi-complete
isomorphism-based ordering.} Lindell's approach~\cite{Lindell1992} for
canonizing trees is based on using a weak order on the class of all trees whose
incomparable elements are exactly the isomorphic ones, and showing that the
order can be computed in logspace. Das, Tor\'{a}n, and Wagner~\cite{DasTW2012}
extended this to also work for graphs with respect to given tree decompositions
of bounded width. This is done by adding the idea that, for bounded width, it is
possible in logspace to guess partial isomorphisms between bags and recursively
check whether they can be extended to isomorphisms between the whole graphs and
the tree decompositions. When working with the tree decompositions into maximal
atoms described above, it is not possible to just guess and check partial
isomorphisms between bags since they have an unbounded width. 

In order to handle the width-unbounded bags of the above decomposition, we use
the fact that (as shown in~\cite{Lokshtanovetal2014a}), after appropriate
preprocessing, the maximal atoms have polynomial-size isomorphism-invariant
families of approximate tree decompositions. To compute these families, we
combine an approach for constructing separator-based tree decompositions
from~\cite{ElberfeldJT2010} to work with the isomorphism-invariant separators
from~\cite{Lokshtanovetal2014a}. If we choose a bounded width tree decomposition
for each atom, and replace each atom by the chosen tree decomposition, we can
turn the width-unbounded decomposition into a width-bounded decomposition for
the whole graph. However, since each maximal atom may be associated with several
decompositions, we need to consider for each atom a family of decompositions. We
call the structure that is obtained a \emph{nested tree decomposition}. In order
to extend the approach that canonizes with respect to width-bounded
decompositions of~\cite{DasTW2012} to nested tree decompositions, we
incorporate a bag refinement step into the weak ordering. It turns root bags of
unbounded width into width-bounded tree decompositions. For each candidate tree
decomposition of the root bag this triggers a modification of the original tree
decomposition. However, it turns out that determining whether there is an
isomorphism between two graphs that respects two given nested tree
decompositions is as hard as the general graph isomorphism problem (see
Remark~\ref{rem:no:exact:correspondence}). Having a polynomial-time algorithm
for this, let alone a logspace algorithm, would thus put the general graph
isomorphism problem into~$\Class{P}$. Consequently, we do not generalize the
idea of using isomorphism-based orderings with respect to decompositions in a
direct way to nested tree decompositions. Instead, we define an approximation of
the isomorphism-based ordering. This approximation has the property that it is
isomorphism-invariant (i.e., graphs that are isomorphic with respect to given
nested decompositions are incomparable) but is only quasi-complete, by which we
mean that graphs that are incomparable must be isomorphic but not necessarily
via an isomorphism that respects the nested decompositions. Developing the
notion of nested tree decompositions along with just the right notion of a
quasi-complete isomorphism-based ordering is our second main technical
contribution.

\emph{Technique 3: Recursive logspace algorithm implementing the quasi-complete
ordering.} Trying all choices of a decomposition on all of the atoms yields
exponentially many refined decompositions in total. Avoiding this exponential
blowup, our third main technical contribution is a dynamic-programming approach
along the tree decomposition that shows how to cycle through candidate
decompositions of the maximal atoms while, still, canonizing the graph along the
coarser tree decomposition in logspace.

Since recursively cycling through tree decompositions of a bag needs space, we
cannot just use the polynomial-size family of tree decompositions that we get
from applying the results of \cite{ElberfeldJT2010} to those of
\cite{Lokshtanovetal2014a} as described above. In order to implement the
recursion in logspace, we compute nested tree decompositions that satisfy a
certain additional (quite technical) property, which we
call~\emph{$p$-boundedness}. It allows us to maintain a trade-off between the
number of candidate tree decompositions chosen for each bag and the size of the
subdecomposition sitting below the bag. This makes a recursive algorithm that
uses only logarithmic space possible.

\paragraph{Organization of the paper.}

Section~\ref{sec:background} provides background on standard graph-theoretic
notions and logspace. The remaining part of the paper is structured along the
proofs of the main theorems: In Section~\ref{sec:decomposition-with-cliques} we
show how to compute isomorphism-invariant tree decompositions into
clique-separator-free graphs in logspace, while
Section~\ref{sec:decomposition-without-cliques} contains the decomposition
approach for graphs without clique
separators. Section~\ref{sec:nested} defines the notion of nested
tree decompositions and a weak ordering that is recursively defined along these
decompositions, while Section~\ref{sec:canonization-ref} proves that the
ordering is logspace-computable for width-bounded and $p$-bounded
decompositions. Section~\ref{sec:main-theorems} finally proves the paper's main
theorems. Section~\ref{sec:conclusion} concludes with a summary and an outlook.

\section{Background}  
\label{sec:background}

We denote the set of natural numbers, which start at 0, by $\Nat$, and use
shorthands $[n,m] := \{n,\dots,m\}$ and $[m] := [1,m]$ for every $n \in \Nat$
and~$m\in \Nat \setminus\{0\}$.

\paragraph{Graphs and Connectivity.} 

For a graph $G = (V,E)$ with \emph{vertices} $V$ and \emph{edges} $E \subseteq V
\times V$, we define $V(G) := V$ and $E(G) := E$. All graphs considered in the
present paper are finite, undirected and simple (neither parallel edges nor
loops are present). We denote the class of all finite graphs by
$\mathcal{G}$. To simplify later definitions, we define the \emph{coloring
function} $\col_G \colon V(G) \times V(G) \to \Int$ of a graph $G$ as
follows. $\col_G(u,v)$ equals $-1$ if $v = w$, $1$ if $v \neq w$ and $\{u,v\}
\in E(G)$, and $0$ if $v \neq w$ and $\{u,v\} \notin E(G)$. If~$G$'s vertices or
edges are \emph{colored}, we extend the coloring function to return natural
number encodings of colors.

\emph{Subgraphs} and \emph{induced subgraphs} are defined as usual. We write
$G[V']$ to denote the subgraph \emph{induced} by a vertex set $V' \subseteq
V(G)$ in a graph $G$. A \emph{path} is an alternating sequence $v_0e_0\dots
v_{m-1}e_{m-1}v_{m}$ of distinct vertices and edges from~$G$ with $e_i =
\{v_i,v_{i+1}\}$ for every $i \in \{0,\dots,m-1\}$.

A \emph{separation} of a graph $G$ is a pair~$(A,B)$ of subsets of $V(G)$ with
(1) $A \cup B = V(G)$, and (2) $(E(G) \cap ((A \setminus B) \times (B \setminus
A)) = \emptyset$. The intersection $A \cap B$ is the \emph{separator} of $(A,B)$
with \emph{size} $|A \cap B|$. The definition of how a separation~$(A,B)$
separates parts of a graph commonly distinguishes between whether the separated
parts are vertices or sets of vertices: A separator $(A,B)$ \emph{separates}
vertex sets $X,Y \subseteq V(G)$ with $X \subseteq A$ and $Y \subseteq B$. A
separator $(A,B)$ \emph{separates} vertices $x,y \in V(G)$ with $x \in A
\setminus B$ and $y \in B \setminus A$. The \emph{connectedness} of sets $X,Y
\subseteq V(G)$ in a graph is the size of a smallest separator that separates
them, it is denoted by $\kappa_G(X,Y)$. The \emph{connectedness} of vertices
$x,y \in V(G)$, denoted by $\kappa_G(x,y)$, is defined in the same way, except
that we set $\kappa_G(x,y) := \infty$ if $x$ and $y$ are adjacent.

\paragraph{Graph isomorphism.}

An \emph{isomorphism} from a (colored) graph $G$ to a (colored) graph $H$ is a
bijective mapping $\phi \colon V(G) \to V(H)$, such that $\col_G(u,v) =
\col_H(\phi(u),\phi(v))$ holds for every $u,v \in V(G)$. Graphs $G$ and $H$ that
admit an isomorphism between them are \emph{isomorphic}. This gives rise to an
equivalence relation that partitions $\mathcal{G}$ into \emph{isomorphism
  classes}. The \emph{graph isomorphism problem} is the language
$\Lang{isomorphism} := \{ (G,H) \in \mathcal{G} \times \mathcal{G} \mid G \text{
  and } H \text{ are isomorphic} \}$. Here an encoding of the graphs as strings is
assumed. (We can assume that the graphs are given as adjacency matrices, which
is however irrelevant since the reasonable encodings are logspace equivalent.)  

A \emph{canonization} is a mapping $\can \colon \mathcal{G} \to
\mathcal{G}$ where $G$ is isomorphic to $\can(G)$, such that for every
two graphs $G$ and $H$ we have $\can(G) = \can(H)$ exactly if $G$ and
$H$ are isomorphic. The graph $\can(G)$ is the \emph{canon of $G$
  (under~$\can$)}, and an isomorphism $\phi$ between~$G$ and $\can(G)$
is a \emph{canonical labeling of $G$ (under~$\can$)}. Comparing the
canons of two graphs $G$ and $H$ suffices to test whether they are
isomorphic, and canonical labelings of two graphs can be used to construct
an isomorphism.

A mapping that associates an object~$\inv(G)$ with every graph~$G \in
\mathcal{G}$, for example a tree decomposition or a family of tree
decompositions, is \emph{isomorphism-invariant} if for every isomorphism $\phi$
between two graphs the result of applying $\phi$ and $\inv$ is independent of
the order in which they are applied. That means, for every isomorphism~$\phi$
from a graph $G$ to a graph $H$, replacing all occurrences of vertices~$v \in
V(G)$ in~$\inv(G)$ by their image~$\phi(v)$ yields~$\inv(H)$. 

\paragraph{Tree decompositions.}

A \emph{(tree) decomposition} $D = (T,\mathcal{B})$ of a graph $G$ is a tree $T$
together with a family of \emph{bags} $\mathcal{B} = (B_n)_{n \in V(T)}$ with
$B_n \subseteq V(G)$ for each $n \in V(T)$, such that
\begin{itemize}
\item (\emph{connectedness property}) for each vertex $v \in V(G)$ the induced
  subtree $T \bigl[\{n \in V(T) \mid v \in B_n\}\bigr]$ is nonempty
  and connected, and
\item (\emph{covering property}) for each edge $\{u,v\} \in E(G)$ there is a
  node $n \in V(T)$ with $\{u,v\} \subseteq B_n$.
\end{itemize} 
For every edge $\{n,m\} \in E(T)$, $B_{\{n,m\}} := B_n \cap B_m$ is the
\emph{adhesion set} between nodes $n$ and $m$. The \emph{torso} of a node $n \in
V(T)$ is the graph obtained from the induced graph $G[B_n]$ by adding for every
neighboring node $n'$ of $n$ the clique on the adhesion set $B_n \cap
B_{n'}$. Given a tree decomposition $D = (T,\mathcal{B})$, its \emph{size} is
$|D| := |V(T)|$, and its \emph{(tree) width} is the maximum over all $|B_n| - 1$
for $n \in V(T)$. The \emph{tree width} of $G$, denoted by $\tw(G)$, is the
minimum width of a tree decomposition for it.

When working with trees underlying \emph{rooted} tree decompositions $D =
(T,\mathcal{B})$, which have a distinguished root node $r \in V(T)$, we talk
about a \emph{parent}, \emph{ancestor}, \emph{child}, and \emph{descendant} of a
node $n \in V(T)$ with respect to the root $r \in V(T)$ in the usual way. Given
a rooted tree decomposition $D = (T,\mathcal{B})$, a
\emph{subdecomposition} $D' = (T',\mathcal{B}')$ is a decomposition that arises
by using a node $n \in V(T)$ and all its ancestor nodes to form a tree
decomposition. A \emph{child decomposition} is a subdecomposition that contains
a child of the root node, but not the root.

Tree decompositions are commonly studied in both their unrooted and rooted
variants. In the context of logspace and the isomorphism problem, we do not need
to restrict ourselves to a particular definition (as formalized by
the following fact). However, in order to facilitate a clear presentation,
we use rooted tree decompositions.

\begin{fact}
  \label{fa:root-decomp} 
  There is a logspace-computable and isomorphism-invariant mapping that turns an
  unrooted tree decomposition $D$ for a graph $G$ into a rooted tree decomposition
  $D'$ for $G$ with the same adhesion sets.
\end{fact}

Fact~\ref{fa:root-decomp} seems to be a folklore and special cases have, for
example, been used in~\cite{Arvindetal2012,Dattaetal2009}. Its proof takes an
unrooted tree decomposition and turns it into a rooted decomposition by
declaring the \emph{center of the tree}, the unique node or edge with the
maximum distance to the tree's leafs, to be the root. If the center is an edge,
we subdivide it by inserting a new node whose bag is the intersection of the
edge's incident bags.

Two graphs $G$ and $G'$ are \emph{isomorphic with respect to tree
decompositions}~$D = (T,\mathcal{B})$ and~$D' = (T',\mathcal{B'})$,
respectively, if there exists an isomorphism $\phi$ from~$G$ to~$G'$ and an
isomorphism~$\psi$ from~$T$ to~$T'$ satisfying $B'_{\psi(n)} = \{ \phi(v) \mid v
\in B_n \}$ for every node~$n \in V(T)$. Under these conditions we say that
$\phi$ \emph{respects} $D$ and $D'$. Based on this definition and the way of how
it refines the isomorphism equivalence relation among graphs, we also consider
\emph{canons of graphs with respect to tree decompositions}.

\paragraph{Logspace.}

A deterministic Turing machine whose working space is logarithmically bounded by
the input length is called a \emph{logspace \dtm{}}. The functions $f \colon
\{0,1\}^* \to \{0,1\}^*$ computed by such machines are
\emph{logspace-computable} (or \emph{in logspace}). The complexity class
$\Class{L}$, called \emph{(deterministic) logspace}, contains all
\emph{languages} $P \subseteq \{0,1\}^*$ whose characteristic functions are in
logspace. Functions in logspace are closed under
composition~\cite{StockmeyerM1973,Jones1975} and also under queries to oracles
for languages from $\Class{L}$~\cite{LadnerL1976}. Reingold~\cite{Reingold2008}
studied the problem $\Lang{undirected-reachability} := \{ (G,s,t) \mid
\textnormal{ there is a path from } s\in V(G) \textnormal { to } t\in V(G)
\textnormal{ in the undirected graph } G\}$, and showed that it is in
$\Class{L}$. Furthermore, we can test whether a graph's tree width is bounded by
a constant since $\Lang{tree-width-$k$} := \{ G \mid \tw(G) \leq k \} \in
\Class{L}$ for every $k \in \Nat$~\cite{ElberfeldJT2010}. Details about the
circuit complexity classes that are mentioned in the introduction are given in
Vollmer's book~\cite{Vollmer1999}, but we do not require them in the following.

\section{Decomposing graphs into parts without clique separators}
\label{sec:decomposition-with-cliques}

A \emph{clique} is a graph with an edge between every two vertices, including
the empty graph by definition. A separation $(A,B)$ is a \emph{clique
separation} with \emph{clique separator} $A \cap B$ in a graph $G$ if it (1)
separates two vertices $x,y \in V(G)$, and (2) $G[A \cap B]$ is a clique.

We construct isomorphism-invariant tree decompositions for graphs of bounded
tree width whose bags induce subgraphs without clique separators and whose
adhesion sets are cliques (that means, the torsos are exactly the subgraphs
induced by the bags). These tree decompositions serve as an intermediate
decomposition step in the proofs of our main theorems.

\begin{lemma}
  \label{lem:main-decomposition-cliques} 
  For every $k \in \Nat$, there is a logspace-computable and
  isomorphism-invariant mapping that turns a graph $G$ with tree width
  at most $k$ into a tree decomposition~$D$ for $G$ in which 
  \begin{enumerate}
  \item subgraphs induced by the bags do not contain clique
    separators, and 
  \item adhesion sets are cliques.
  \end{enumerate}
\end{lemma}

The tree decomposition we construct to prove the lemma is a refined version of a
decomposition of Leimer~\cite{Leimer1993} of graphs into their collections of
maximal induced subgraphs without clique separators. The crucial point is that
we need to adjust his method to not only output the collection of maximal
induced subgraphs without clique separators, which suffices for its application
in~\cite{Lokshtanovetal2014a}, but also an isomorphism-invariant tree
decomposition that is based on it. In order to do that, we replace the approach
of~\cite{Leimer1993}, which is based on finding clique-separator-free parts in a
single phase via computing elimination orderings, by several steps. In these
steps we compute graphs that are clique-separator-free with respect to cliques
up to a certain size. The size bound grows when going from one step to the
next. While Leimer's method runs in polynomial-time and applies to general
graphs, our approach needs logarithmic space and applies to graphs of bounded
tree width, which suffices for our applications.

Section~\ref{sec:catoms} presents the definition and logspace-computability of
maximal subgraphs without size-bounded clique
separators. Section~\ref{sec:chordal} presents a transformation of graphs that
is used in Section~\ref{sec:catoms-decomposition} to prove
Lemma~\ref{lem:main-decomposition-cliques}

\subsection{Atoms with respect to constant-size clique separators}
\label{sec:catoms}

Let $c \in \Nat$, which we use as an upper bound on the size of
clique separators we consider. A \emph{$c$-atom} is a graph that does
not contain clique separators of size at most~$c$. \emph{Atoms}, which
are the graphs Leimer~\cite{Leimer1993} deals with, are $|V(G)|$-atoms
(or, alternatively, graphs without clique separators). A \emph{maximal
$c$-atom} of a graph~$G$ is a maximal induced subgraph $G[A]$ for some
$A \subseteq V(G)$ that is a $c$-atom (that means, every extension of
it contains a clique separator of size at most~$c$). A \emph{maximal
atom} in a graph $G$ is a maximal $|V(G)|$-atom. For every $c \in
\Nat$, $c$-atoms are nonempty and connected.

Two vertices~$a_1,a_2 \in V(G)$ are \emph{$c$-inseparable (with
respect to clique separations in $G$)} if there is no clique separator
of size at most~$c$ in $G$ separating $a_1$ and $a_2$ and
\emph{$c$-separable}, otherwise. A set~$A \subseteq V(G)$
is~\emph{$c$-inseparable} if all distinct $a_1,a_2 \in A$ are
$c$-inseparable. The set $A$ is \emph{maximal $c$-inseparable} if no extension
of $A$ is $c$-inseparable.

Note that the definition of induced subgraphs $G[A]$ that are $c$-atoms only
considers separations of~$G[A]$ while $c$-inseparability of a vertex
set $A$ is based on separations in the (ambient) graph~$G$. If we look
at maximal $c$-atoms and maximal $c$-inseparable sets, then these
notions coincide.

\begin{lemma}
  \label{le:atom-inseparable-set} 
  For every~$c\in \Nat$, graph $G$, and~$A\subseteq V(G)$, $G[A]$ is a
  maximal~$c$-atom of~$G$ if and only if~$A$ is maximal $c$-inseparable in
  $G$. 
\end{lemma}
\begin{proof}
(From $c$-atoms to $c$-inseparable sets.) Two vertices $a_1$ and $a_2$ from a
vertex set $A$ that are $c$-separable in $G$ are also $c$-separable in
$G[A]$. This implies that, if $G[A]$ is a $c$-atom (meaning that $A$ is
$c$-inseparable in~$G[A]$), then $A$ is $c$-inseparable in $G$.

(From $c$-inseparable sets to $c$-atoms.) Let~$A \subseteq V(G)$ be
maximal~$c$-inseparable. Assume, for the sake of contradiction, that~$a_1,a_2
\in A$ are separated by a clique separator~$C$ of size at most~$c$
in~$G[A]$. Since~$C$ does not separate~$a_1$ from~$a_2$ in~$G$, there is a path
from~$a_1$ to~$a_2$ in~$G$ that avoids~$C$. Let~$P$ be a shortest path of this
kind. Since~$C$ separates~$a_1$ from~$a_2$ in~$G[A]$, there is an~$x \in V(G)
\setminus A$ on the path~$P$. Since~$x \notin A$, but~$A$ is chosen to be
maximal $c$-inseparable, there is a vertex~$a' \in A$ and a clique
separator~$C'$ of size at most $c$ that separates~$a'$ from~$x$. Since~$A$ is
$c$-inseparable in $G$,~$C'$ cannot separate elements of~$A$ in~$G$. Thus, we
know that either~$C'$ separates~$x$ from~$a_1$ or~$a_1 \in C'$, and either~$C'$
separates~$x$ from~$a_2$ or~$a_2 \in C'$. That means, the path~$P$, which starts
in~$a_1$, passes through~$x$, and ends in~$a_2$, must intersect~$C'$ in some
vertex~$p_1 \neq x$ before reaching~$x$, and intersect~$C'$ again in some
vertex~$p_2 \neq x$ after leaving~$x$. Since~$C'$ is a clique, we can take a
shortcut by directly taking the edge~$\{p_1,p_2\} \in E(G)$ without
visiting~$x$. This contradicts the fact that~$P$ is a shortest path.
\end{proof}

\begin{lemma}
  \label{le:inseparable-unique}
  Let $G$ be a graph with $c$-inseparable set $I \subseteq V(G)$ with
  $|I| \geq c+1$ for some $c \in \Nat$. Then $A := \{a \in V(G) \mid I
  \cup \{a\} \text{ is $c$-inseparable} \}\,$ is the unique maximal
  $c$-inseparable set in $G$ with $I \subseteq A$.
\end{lemma}
\begin{proof} 
  We first argue that~$A$ is~$c$-inseparable: Assume, for
  sake of contradiction, that vertices $a_1,a_2 \in A$ are~$c$-separable
  in $G$ via a clique separator $C$ of size at most $c$. Since~$|C| \leq
  c < c + 1 \leq |I|$, there is a vertex~$v \in I \setminus
  C$. Moreover, since~$C$ separates~$a_1$ from~$a_2$, $C$
  separates~$a_1$ from~$v$ or it separates~$a_2$ from~$v$. This
  contradicts~$a_1,a_2 \in A$ since both $I \cup \{a_1\}$ and $I \cup
  \{a_2\}$ are $c$-inseparable in $G$ by construction. To see that~$A$
  is unique and maximal among the~$c$-inseparable sets containing~$I$,
  note that every candidate~$a \in V(G) \setminus A$ is~$c$-separable
  from at least one vertex in~$I$.
\end{proof}

Combining Lemma~\ref{le:atom-inseparable-set} with
Lemma~\ref{le:inseparable-unique}, we are able to compute the family
of $c$-atoms of a given graph for every constant $c \in \Nat$ in
logspace.

\begin{lemma}
  \label{lem:atoms:logspace:computable}
  For every~$c\in \Nat$, the mapping that turns a graph $G$ into
  its family of maximal $c$-atoms is logspace-computable and
  isomorphism-invariant. 
\end{lemma}
\begin{proof} 
  We show how to compute the family of all maximal sets
  $A$ that are $c$-inseparable in~$G$. By
  Lemma~\ref{le:atom-inseparable-set}, these are exactly the vertex sets
  of the maximal $c$-atoms of $G$. We first check for all sets~$A$ of
  size at most~$c$ whether they are maximal $c$-inseparable in~$G$, and
  output the sets that pass the test. To find maximal~$c$-inseparable
  sets with more than~$c$ vertices, we consider every $c$-inseparable
  set~$I \subseteq V(G)$ with $|I| = c+1$ and output $A := \{a \in V(G)
  \mid I \cup \{a\} \text{ is inseparable} \}$. The correctness of the
  algorithm follows from Lemma~\ref{le:inseparable-unique}. It can be
  implemented by a logspace \dtm{} since we only cycle through vertex
  sets of (constant) size at most $c$ and use oracle calls to
  $\Lang{undirected-reachability}$, which is in
  $\Class{L}$~\cite{Reingold2008}.

  The mapping that turns a graph $G$ for a $c \in \Nat$ into
  its family of $c$-inseparable sets is isomorphism-invariant by
  definition.   
\end{proof}

\subsection{Chordal completions with respect to constant-size clique separators} 
\label{sec:chordal}

Instead of working with a given graph~$G$ directly, some of Leimer's
arguments~\cite{Leimer1993} are based on working with its \emph{chordal
completion}~$G^*$, which is the graph that arises from~$G$ by replacing every
maximal atom in the graph with a clique on the atom's vertices. Interestingly,
the vertex sets of the maximal atoms are the same for $G$ and $G^*$. The same
property holds when moving from a graph~$G$ to the \emph{$c$-chordal completion}
$G^c$, for every $c \in \Nat$, that arises from $G$ by replacing every maximal
$c$-atom with a clique on its vertex set. For a formal proof of this fact we
first show that the intersection of two~$c$-atoms is a clique.

\begin{lemma}
  \label{le:inseparable-intersection}
  Let $G$ be a graph and $A_1$ and $A_2$ two distinct maximal
  $c$-inseparable sets in $G$. Then $A_1 \cap A_2$ is a clique of 
  size at most $c$.
\end{lemma}
\begin{proof}
  Since $I = A_1 \cap A_2$ is $c$-inseparable, it has size at most~$c$
  in order to be contained in two distinct maximal $c$-inseparable
  sets by Lemma~\ref{le:inseparable-unique}. Let $a_1 \in A_1
  \setminus A_2$ and $a_2 \in A_2 \setminus A_1$ be $c$-separable
  vertices and $C$ be a clique separator of size at most $c$ that
  separates them. Assume, for the sake of contradiction, that there
  exists a vertex $a' \in (A_1 \cap A_2) \setminus C$. Since both $a'
  \in A_1$ and $a' \in A_2$ and both $A_1$ and $A_2$ are
  $c$-inseparable, we can find a path $P_{a_1,a'}$ between $a_1$ and
  $a'$ and a path $P_{a',a_2}$ between $a'$ and $a_2$ in $G -
  C$. Thus, there is also a path $P$ between $a_1$ and $a_2$. This
  contradicts the existence of $a' \in (A_1 \cap A_2) \setminus
  C$. Thus, the clique $C$ contains all of $A_1 \cap A_2$. In
  particular, $A_1 \cap A_2$ is a clique. 
\end{proof}

\begin{lemma}
  \label{le:maximal-atom-G-Gc}
  Let $c \in \Nat$, $G$ be a graph, $A \subseteq V(G)$, and $C
  \subseteq V(G)$ with $|C| \leq c$. Then
  \begin{enumerate}
  \item $G[A]$ is a maximal $c$-atom in $G$ if and only if $G^c[A]$ is
    a maximal $c$-atom in $G^c$, and  
  \item $C$ is a clique separator in $G$ if and only if $C$
    is a clique separator in $G^c$
\end{enumerate}  
\end{lemma}
\begin{proof} 
  We start to prove the first property. The arguments of the proof are
  based on properties of vertex sets of maximal $c$-atoms, and maximal
  $c$-insep\-arable vertex sets. We can freely switch between both
  points of view since they are equivalent by
  Lemma~\ref{le:inseparable-unique}.

  Let $A \subseteq V(G)$ be maximal $c$-inseparable in $G$. Since
  $G^c[A]$ is a clique, $A$ is also $c$-inseparable in $G^c$. To prove
  that $A$ is maximal with this property in $G^c$, assume, for the sake of
  contradiction, that $A$ is not maximal with this property in
  $G^c$. Then there is a vertex $x \in V(G) \setminus A$, such that $A
  \cup \{x\}$ is $c$-inseparable in $G^c$. Since $A$ is maximal
  $c$-inseparable in $G$, there is a clique separator $C$ of size at
  most $c$ that separates a vertex $a \in A$ from $x$ in $G$. By
  assumption, $C$ does not separate $a$ from $x$ in $G^c$. Thus, we can
  find a path $P$ between $a$ and $x$ in $G^c - C$. We show that $P$ can
  be modified to a path $P'$ between $a$ and $x$ in $G$ by observing
  each of its edges and, if necessary, redirecting it. Let $\{u,v\}$ be
  an edge from $P$. If $\{u,v\} \in E(G)$, we are done. If $\{u,v\} \in
  E(G^c) \setminus E(G)$, we know by the construction of $G^c$ that $u$
  and $v$ are $c$-inseparable in $G$ (in particular, they are part of
  a common maximal $c$-inseparable set). Thus, we can find a path
  $P_{\{u,v\}}$ between $u$ and $v$ in $G[A'] - C$ and modify~$P$ to
  use $P_{\{u,v\}}$ instead of $\{u,v\}$, which is only present
  in~$G^c$. Overall, this leads to constructing a path $P'$ in~$G - C$
  between $a$ and $x$. This contradicts the fact that $C$ separates
  $a$ and $x$ in $G$.  

  We are left to prove the converse direction. Let $a_1$ and $a_2$ be two
  vertices that are $c$-inseparable in~$G^c$. For the sake of contradiction,
  assume they are $c$-separable in $G$. Let $(B_1,B_2)$ be a clique separation in
  $G$ with clique separator $C = B_1 \cap B_2$ of size at most $c$ and $a_1 \in
  B_1 \setminus B_2$ and $a_2 \in B_2 \setminus B_1$. In particular, this means
  that every vertex $a'_1 \in B_1 \setminus B_2$ is $c$-separable from every
  vertex $a'_2 \in B_2 \setminus B_1$. Hence, no edges are constructed between the
  sets $B_1 \setminus B_2$ and $B_2 \setminus B_1$, and $C$ is also a clique
  separator of size at most $c$ in $G^c$. This contradicts the initial choice of
  $a_1$ and $a_2$ as being $c$-inseparable vertices in $G^c$. Thus, $a_1$ and
  $a_2$ are $c$-inseparable in $G$, too. The arguments above imply that every
  $c$-inseparable set of $G^c$ is a $c$-inseparable set in $G$. In particular,
  this holds for the maximal $c$-inseparable sets and, thus, for the maximal
  $c$-atoms.

  For the second property, let $C$ be a clique separator of size at
  most $c$ in $G$, which does not need to be a minimum clique
  separator. Then $C$ is also a clique separator in $G^c$ by the
  arguments from the last paragraph. For the other direction, let $C$
  be a clique separator of size at most $c$ in $G^c$. It also
  separates two distinct vertices in $G$ and we are left to prove that
  it is a clique. Since $C$ is a clique separator of size at most $c$,
  it lies in the intersection of two distinct maximal $c$-inseparable 
  sets $A_1$ and $A_1$ in $G^c$. The sets $A_1$ and $A_1$ are also
  maximal $c$-inseparable sets of $G$ by the first property of the
  lemma proved above, and $A_1 \cap A_2$ is a clique by
  Lemma~\ref{le:inseparable-intersection}. Thus, $C$ is a clique
  separator in $G$, too.
\end{proof}

\subsection{Isomorphism-invariant tree decompositions into atoms}  
\label{sec:catoms-decomposition}

Our goal is to compute an isomorphism-invariant tree decomposition of a graph
into its $c$-atoms.

A \emph{minimum clique separator (with respect to $x$ and $y$ in a graph $G$)}
is an inclusion-wise minimal clique that separates $x$ and $y$ in $G$. For every
$c \in \Nat$ and graph~$G$, we define the graph $T_c = T_c(G)$ whose node set consists of all $c$-atoms of $G$ and all minimum clique separators of size at
most~$c$. An edge is inserted between every $c$-atom~$G[A]$ and minimum clique
separator $C$ with $C \subseteq A$. We define the class of bags
$\mathcal{B}_c(G) = (B_n)_{n \in V(T_c)}$ as follows. If $n \in V(T_c)$ is
identified with a $c$-atom $G[A]$, then $B_n := A$, and if $n \in V(T_c)$ is
identified with a minimum clique separator $C$, then $B_n := C$.

\begin{proposition}
  For every $c \in \Nat$, the mapping $G \mapsto (T_{c}(G),\mathcal{B}_c(G))$ is
  isomorphism-invariant for every $c \in \Nat$.
\end{proposition}

The graph~$T_c(G)$ is typically not a tree. However, as stated by
Lemma~\ref{lem:c-atom-tree}, $T_c(G)$ is a tree and, moreover,
$(T_c(G),\mathcal{B}_c(G))$ is a tree decomposition if $G$ is a
$(c-1)$-atom.

\begin{lemma}
  \label{lem:c-atom-tree}
  For every positive $c \in \Nat$ and $(c-1)$-atom $G$,
  $(T_c(G),\mathcal{B}_c(G))$ is a tree decomposition for
  $G$. Moreover, $T_c(G)$ has a unique center.
\end{lemma}
\begin{proof}
  Instead of working with $G$, we use the graph $G^c$, whose maximal
  $c$-inseparable vertex sets and clique separators of size at most $c$
  are exactly the respective ones of $G$ by
  Lemma~\ref{le:maximal-atom-G-Gc}. Thus, we set $G := G^c$ throughout
  the proof, which does not alter the construction of $T_c$. To simplify
  the notations of the proof, we also set $T := T_c(G)$ and $\mathcal{B}
  := \mathcal{B}_c(G)$. 

  \begin{claim*} 
    $T$ is connected. 
  \end{claim*}
  \begin{proof}[Proof of the claim]
    \renewcommand{\qedsymbol}{\hfill$\lrcorner$}
    If~$T$ is a single node, the claim holds. If $V(T) \geq
    2$, we argue as follows.

    Since every clique separator is contained in some
    maximal~$c$-atom, it suffices to show that distinct maximal~$c$-atoms
    are connected in~$T$.  Let~$A_1$ and~$A_2$ be
    maximal~$c$-atoms. Since~$A_1$ and~$A_2$ are distinct, there is a
    clique separator~$C$ of size at most~$c$ separating~$A_1$ from~$A_2$
    in $G$. For an atom~$A$ and a clique separator~$C$ separating~$A$ from
    a vertex~$x\in V(G)\setminus A$, define~$\Delta(A,C)$ to be the
    minimum~$\sum_{i=1}^c |P_i|$ among all~$c$-tuples of vertex-disjoint
    paths~$P_1,\ldots,P_c$ that start in~$A$ and end in~$C$. Such paths
    exist by Menger's theorem~\cite{Diestel2005} after the preprocessing
    mentioned above. Note that~$\Delta(A,C) = 0$ if and only
    if~$C\subseteq A$.  We show that distinct atoms~$A_1$ and~$A_2$ are
    connected in~$T$ by induction
    on~$\Delta(A_1,C)+\Delta(A_2,C)$. If~$\Delta(A_1,C) = \Delta(A_2,C) =
    0$, then~$C\subseteq A_1$ and~$C\subseteq A_2$ so~$A_1$ and~$A_2$ are
    connected in~$T$ by definition. Thus, we assume that for all
    cliques~$C$ of size~$c$ separating~$A_1$ and~$A_2$ we
    have~$\Delta(A_1,C) + \Delta(A_2,C)>0$. To continue the proof, we
    distinguish two cases.

    For the first case, assume there exists a
    clique~$C$ of size~$c$ separating~$A_1$ and~$A_2$ with~$\Delta(A_1,C)>0$
    and~$\Delta(A_2,C)>0$. This implies both~$C\nsubseteq A_1$ and~$C\nsubseteq
    A_2$. Since~$C$ is~$c$-inseparable, there is a~$c$-atom~$A'$ that
    contains~$C$.  Since every path from~$A_1$ to~$A_2$ must
    intersect~$C$, we have $\Delta(A',C)+\Delta(A_i,C) < \Delta(A_1,C)+
    \Delta(A_2,C)$ for~$i\in\{1,2\}$. Thus, $A'$ is connected to $A_i$
    in~$T$ for every $i \in \{1,2\}$ and, thus, $A_1$ is connected to~$A_2$. 
    
    For the second case, suppose that for all clique separators~$C$ of
    size~$c$ separating~$A_1$ and~$A_2$, we have~$\Delta(A_1,C) =0$,
    or~$\Delta(A_2,C)=0$. Let~$C$ be such a clique separator. Without loss of
    generality, we may assume~$\Delta(A_1,C) = 0$. Since~$ \Delta(A_2,C) >
    0$, there is an element~$v\in C\setminus A_2$. Let~$C'$ be a clique
    that separates~$v$ from~$A_2$. Since~$v\in C \setminus A_2$, we
    conclude that~$\Delta(A_2,C')=0$ and, thus,~$C' \subseteq A_2$. If
    there is an atom~$A'$ containing~$C$ and~$C'$, then~$A'$ is adjacent
    to~$A_1$ and~$A_2$ and, thus,~$A_1$ and~$A_2$ are
    connected. Otherwise, there must be a clique separator~$C''$
    separating a vertex in~$C$ from a vertex in~$C'$. However, this
    implies~$C''\nsubseteq A_1$ and~$C''\nsubseteq A_2$. This brings us to
    the previous case with~$\Delta(A_1,C'')> 0$ and~$\Delta(A_2,C'')>0$.
  \end{proof}

  \begin{claim*}
    If~$A_1,C_1,A_2\ldots,C_{t-1},A_t$ is a path in~$T$ between two
    $c$-atoms, then for every~$i\geq 2$, the $c$-atom~$A_i$ does not contain a
    vertex from~$A_1\setminus A_2$. 
  \end{claim*}
  \begin{proof}[Proof of the claim]
    \renewcommand{\qedsymbol}{\hfill$\lrcorner$} Let~$v_1$ be a vertex
    in~$A_1\setminus A_2$. Such a vertex must exists since~$A_1$ and~$A_2$
    are distinct maximal~$c$-atoms.  We show for~$2 \leq i \leq t-1$ that
    the separator~$C_i$ contains a vertex~$v_i$ that is separated
    from~$v_1$ by~$C_1$. For $i \geq 2$, we choose~$v_i$ as a vertex
    in~$C_{i}\setminus C_1$. Such a vertex exists since the~$C_i$ are
    distinct subsets of~$V(G)$ of the same size~$c$. To see that~$C_1$
    separates~$v_1$ from~$v_i$, observe that by induction~$C_1$
    separates~$v_1$ from~$v_{i-1}$, but~$v_{i-1}$ and~$v_i$ cannot be
    separated by~$C_1$ since they lie in the same
    $c$-atom~$A_i$. Since~$v_{t-1} \in A_t$,~$C_1$ separates all vertices
    in~$A_1\setminus A_2$ from all vertices in~$A_t \setminus A_2$. This
    proves the claim.
  \end{proof}

  To see that~$T$ is a tree, assume that $A_1,C_1,A_2\ldots,C_{t-1},A_t=
  A_1$ is a cycle. By the second claim, $A_t = A_1$ does not contain a vertex
  from~$A_1\setminus A_2$, but this contradicts~$A_1$ and~$A_2$ being maximal
  $c$-atoms that are distinct. To see that the center of~$T$ is a unique
  node, it suffices to observe that a separator cannot be a leaf of~$T$.

  To prove the connectedness property of decompositions, let
  $A_1,C_1,A_2,\ldots\allowbreak,C_{t-1},A_t$ be a path, such that~$A_1$
  and~$A_t$ contain a common vertex~$v$ that is not contained in~$A_i$
  for~$i\in \{2,\ldots,t-1\}$. Since separators are always contained in
  some adjacent atom, this is the only case that needs to be
  considered. However, the existence of such a path directly contradicts
  the second claim above. The covering property of tree decompositions
  follows from the fact that every edge of $G$ is part of some $c$-atom. 
  Hence, $T$ is a tree decomposition of~$G$.
\end{proof}

\begin{lemma}
  \label{lem:c-to-c++-atom-decomposition} 
  For every $d,c \in \Nat$, with~$d < c$, there is a logspace \dtm{} that, given a
  $d$-atom~$G$, outputs a tree decomposition $D = (T,\mathcal{B})$ 
  \begin{enumerate}
  \item whose bags are $c$-atoms,
  \item whose adhesion sets are cliques, and
  \item where the mapping is isomorphism-invariant. 
  \end{enumerate}
\end{lemma}
\begin{proof}
  We show the lemma by induction on~$c-d$. If~$c-d = 0$, then the
  $c$-atom $G$ is the unique bag of the tree decomposition $D$, which
  satisfies all requirements of the lemma. If~$c-d > 0$, we construct
  and prove the correctness of the constructed tree decomposition as
  follows.

  (Construction.) We use Lemma~\ref{lem:c-atom-tree} to construct a tree
  decomposition~$D' = (T',\mathcal{B}')$ whose bags are
  the graph's $(d+1)$-atoms and minimum clique separators. Applying the
  induction hypothesis, we compute for each $(d+1)$-atom~$A$ an
  isomorphism-invariant tree decomposition~$D_A = (T_A,\mathcal{B}_A)$ into
  its~$c$-atoms. We continue combining~$D'$ with the decompositions $D_A$ to
  construct $D = (T,\mathcal{B})$. We use nodes $V(T) := \{(B,A) \mid
  B \in V(T_A) \text{ and } A \in V(T') \}$ 
  for $T$. Two nodes~$(B_1,A_1)$ and~$(B_2,A_2)$ of~$T$ are adjacent
  if (1) $A_1 = A_2$ and~$B_1$ and~$B_2$ are adjacent in~$D_{A_1}$, or
  (2) $A_1$ and~$A_2$ are adjacent in~$T$ and for each~$i \in \{1,2\}$, $B_i$
  contains $A_1 \cap A_2$ and is closest to the root with this property in
  $D_{A_i}$. To each node of $(B,A)$ of $T$, $\mathcal{B}$ assigns the
  bag $B_{(B,A)} = B$. 

  Since distances in trees are logspace-computable, we can determine the bag
  closest to the root in the definition above. Thus, $T$ can be constructed in
  logspace based on constructing $T$ by Lemma~\ref{lem:c-atom-tree} and $T_A$ by
  induction.

  (Correctness of construction.)
  The tree~$T$ is well defined since the intersection of two atoms~$A_1$
  and~$A_2$ that are adjacent in~$D'$ is a clique and every clique must be
  contained in some bag of a tree decomposition. Moreover, the bags that contain a
  clique form a connected subtree and the bag closest to the root is well defined.

  Isomorphism invariance of~$T$ follows from the isomorphism invariance of the
  decompositions $D_A$, $D'$, and the uniqueness of the bag closest to the root
  containing a clique.

  We argue that~$T$ is a tree. It is connected since $D'$ and each $D_A$ is
  connected. To argue that it is cycle free, suppose~$(B_1,A_1),\ldots,(B_t,A_t)$
  with $(B_t,A_t) = (B_1,A_1)$ is a cycle. Note that for two atoms~$A$ and $A'$
  that are adjacent in~$D'$ the bag~$B$ and~$B'$ for which~$(B,A)$ is adjacent
  to~$(B,A')$ is unique. This implies that either the walk~$A_1,\ldots,A_t$
  contains a cycle, or there are indices~$1\leq j<k\leq t$, such that~$A_{j}=
  A_{j+1} = \dots = A_{k}$ and~$(B_{j},A_{j} ), (B_{j+1},A_j ), \ldots, (B_{k},A_j
  )$ is a closed walk, which implies that~$B_{j}, B_{j+1},\ldots,B_{t}$ is a
  cycle. Since both~$D'$ and~$D_{A}$ are acyclic, this yields a contradiction.

  It remains to show that with this definition~$T$ is a tree decomposition whose
  adhesion sets are cliques. If two vertices~$(B_1,A_1)$ and~$(B_2,A_2)$ are
  adjacent, then $B_1\cap B_2$ is an adhesion set either in~$D_{A_1} = D_{A_2}$,
  or in~$D$. In either case it is a clique. To show the connectivity property of
  tree decompositions, let~$(B_1,A_1),\ldots,(B_t,A_t)$ be a path in~$T$ such
  that~$B_1$ and~$B_t$ contain a vertex~$v$ that does not appear in~$B_i$
  for~$2\leq i\leq t-1$. This implies, since~$D'$ is a tree decomposition,
  that~$v$ is contained in all~$A_i$. In turn, this implies that~$v$ is contained
  in all~$B_i$, since each~$D_A$ is a tree decomposition.
\end{proof}

\begin{proof}[Proof of Lemma~\ref{lem:main-decomposition-cliques}]
  Let~$G$ be the input graph of tree width at most~$k$. Without loss of
  generality, we assume that it is connected. Then the graph $G$ is a
  $0$-atom. We apply Lemma~\ref{lem:c-to-c++-atom-decomposition} to $G$
  with $d = 0$ and $c = k+1$. Since $G$ has tree width at most $k$, the
  size of a largest clique in $G$ is bounded by $k+1$. Thus, the
  subgraphs induced by the bags, which do not contain clique separators
  of size at most $k+1$ by their construction, do not contain clique
  separators (of any size).
\end{proof}

\section{Decomposing graphs without clique separators}
\label{sec:decomposition-without-cliques}

The decomposition procedure from the previous section provides us with a tree
decomposition whose bags are clique-separator-free. In the present section, we
decompose clique-separator-free graphs further into isomorphism-invariant tree
decompositions of bounded width (formalized by
Lemma~\ref{lem:decomp-without-clique}). This needs two additional assumptions
that we later meet during the proofs of Theorems~\ref{th:isomorphism-tw}
and~\ref{th:canonization-tw}. First, the decomposition is based on two
distinguished nonadjacent vertices from the graph. Second, we assume that the
given graph is improved as defined next.

Let $\impr \colon \mathcal{G} \to \mathcal{G}$ be the mapping that takes a graph
$G$ and adds edges between all vertices $u,v \in V(G)$ with $\kappa(u,v) >
\tw(G)$. The $\impr$-operator \emph{improves} the graph by adding edges of $G$
based on its tree width. To avoid losing information, we introduce a function
$\col_{\impr(G)}$ that colors edges that appear originally in the inputs with a
different color than those coming from the improvement. The mapping $\impr$ is
isomorphism-invariant by definition. Besides this, we use three further
properties of the mapping $\impr$. First, the graph we get from applying $\impr$
is saturated in the sense that a second application of it does not add new
edges. Formally, this means $\impr(G) = \impr(\impr(G))$ for every graph $G$ as
proved in \cite[Lemma 2.5]{Lokshtanovetal2014b}. Second, the tree decompositions
of a graph~$G$ are exactly the tree decompositions of~$\impr(G)$. This implies
$\tw(G) = \tw(\impr(G))$ and is proved in \cite[Lemma
2.6]{Lokshtanovetal2014b}. Third, the mapping $\impr$ is logspace-computable for
graphs of bounded tree width. This follows from Reingold's algorithm for
$\Lang{undirected-reachability}$, and the fact that the tree width of a graph
bounds the size of the separators we need to consider in order to compute
$\impr$.

\begin{lemma}
  \label{lem:decomp-without-clique} 
  For every $k \in \Nat$, there is a $k' \in \Nat$ and a logspace-computable and 
  isomorphism-invariant mapping that turns every graph $G$ with a distinguished non-edge $\{u,v\} \notin
  E(G)$, where $G$ 
  \begin{enumerate}
  \item has tree width at most $k$,
  \item does not contain clique separators, and
  \item is improved (that means, $G = \impr(G)$),
  \end{enumerate}
  into a width-$k'$ tree decomposition $D = (T,\mathcal{B})$ for $G$. 
\end{lemma}

The rest of the section is devoted to the proof of the lemma.
The construction of the decomposition is based on recursively splitting the
graph into smaller subgraphs using size-bounded and
isomorphism-invariant separators. In order to do this, we adapt in a first step
the isomorphism-invariant separators from~\cite{Lokshtanovetal2014a}
and show their logspace-computability (this is done in
Section~\ref{sec:without-cliques-separators}).  Then we combine this
with a logspace approach for handling the recursion involved in this
approach from~\cite{ElberfeldJT2010} (this is done in
Section~\ref{sec:construct-decomposition-without-cliques}).

\subsection{Constructing isomorphism-invariant separators} 
\label{sec:without-cliques-separators}

Lokshtanov et al. \cite{Lokshtanovetal2014a} identified an isomorphism-invariant
family of separators that can be used as part of a recursive algorithm for
constructing isomorphism-invariant tree decompositions. We adapt their approach
of constructing separators, which is tailored to find (time-efficient)
algorithms proving fixed-parameter tractability, to work in logspace. For this
we need to adjusted some terminology.

For a graph $G$ and a set of vertices $V' \subseteq V(G)$, we define the
neighborhood of $V'$ in $G$ to be the set $N_G(V') := \{
v \in V(G) \setminus V' \mid \text{there exists } w \in V' \text{ with } \{v,w\}
\in E(G)\}$.
A \emph{graph with interface} is a pair $(G,I)$ consisting of a graph $G$ and an
\emph{interface} $I \subseteq V(G)$ where    
\begin{enumerate}
\item $G \setminus I$ is connected, and
\item $I = N_G(V(G) \setminus I)$.
\end{enumerate}
We split a graph into several components based on separators for its
interface. Let $G$ be a graph and $X,Y \subseteq V(G)$. It is well-known (see,
for  example, \cite{Lokshtanovetal2014a}) that there is a unique separator $(A,B)$
for $X$ and $Y$ of minimum size and (inclusion-wise) minimal $A$. We denote it
by $(A_X,B_Y)$ and set $\operatorname{sep}(X,Y) := A_X \cap B_Y$. Exactly the
same property holds when considering vertices $x,y \in V(G)$. In this case we
denote the corresponding separation by $(A_x,B_y)$, and set
$\operatorname{sep}(x,y) := A_x \cap B_y$. We use these separations to define
the \emph{separator} $\operatorname{sep}_s(G,I)$ of a graph with interface
$(G,I)$ with respect to a threshold value $s \in \Nat$: If $|I| \leq s$, we set  
\begin{align*} 
  \operatorname{sep}_s(G,I) := I \cup
  \bigcup_{\substack{ x,y \in I, \, x \neq y, \textnormal{ and} \\
  \kappa_G(x,y) \leq \tw(G)}} \operatorname{sep}(x,y) , \textnormal{ and }
\end{align*}
\begin{align*} 
  \operatorname{sep}_s(G,I) := I \cup \bigcup_{\substack{ X,Y \subseteq
  I, \, X \cap Y = \emptyset, \, |X| = |Y| = \tw(G) + 1, \textnormal{
  and} \\ \kappa_G(X,Y) \leq \tw(G)}} \operatorname{sep}(X,Y), \textnormal{
  otherwise.}
\end{align*}

The following proposition follows from the definition of
$\operatorname{sep}_s(G,I)$, and the constant bound on the tree width
of the given graphs. In this case, in order to compute
$\operatorname{sep}_s(G,I)$, we only need to enumerate vertex sets of
constant size combined with reachability queries in undirected
graphs. Moreover, for $s \in \Nat$, we know that $(G,I) \mapsto
\operatorname{sep}_s(G,I)$ is isomorphism-invariant by definition.

\begin{proposition}
  \label{pr:separators-logspace}
  For every $k \in \Nat$, there is a logspace \dtm{} that, given a graph with
  interface $(G,I)$ where $\tw(G) \leq k$ and $s \in \Nat$, outputs
  $\operatorname{sep}_s(G,I)$. 
\end{proposition}

The following fact on size bounds for separators and neighborhoods follows from
the statements and proofs of \cite[Lemmata 3.3 and 3.4]{Lokshtanovetal2014b}.

\begin{fact}
  \label{fa:without-cliques-size-bounds} 
  There are functions $\operatorname{small} \in O(k)$,
  $\operatorname{medium} \in O(k^3)$, and $\operatorname{large} \in 
  O(2^{k \log k})$ with the following properties: Let $(G,I)$ be a graph with
  interface, such that $\tw(G) \leq k$, $G$ is improved and an atom such that
  \begin{enumerate}
  \item $G[I]$ is not a clique, and 
  \item $|I| \leq \operatorname{medium}(\tw(G))$.
  \end{enumerate}
  Moreover, let $S :=
  \operatorname{sep}_{\operatorname{small}(\tw(G))}(G,I)$. Then $S
  \setminus I \neq \emptyset$, $|S| \leq \operatorname{large}(\tw(G))$, and for
  every component~$C_i$ of $G \setminus S$ with its graph with interface
  $(G_i,I_i) := (G[V(C_i) \cup N_G(V(C_i))],N_G(V(C_i)))$
  \begin{enumerate}
  \item $G[I_i]$ is not a clique, and
  \item $|I_i| \leq \operatorname{medium}(\tw(G))$.
  \end{enumerate}
\end{fact}

\subsection{Constructing isomorphism-invariant tree decompositions}
\label{sec:construct-decomposition-without-cliques}

To construct iso\-mor\-phism-in\-variant tree decompositions using the
previously defined (isomorphism-invariant) separators for graphs, we encapsulate
their recursive computation using the concept of descriptor decompositions
from~\cite{ElberfeldJT2010}. We slightly adjust the terminology
from~\cite{ElberfeldJT2010} by using graphs with interfaces directly instead of
using descriptors.

A \emph{descriptor decomposition} for a graph $G$ is a pair $(M,\mathcal{R})$
consisting of a directed graph $M$ and a collection $\mathcal{R}$ of subgraphs
with interfaces $R_n = (H,I)$ for every node $n \in V(M)$ where $(V(G) \setminus
(V(H) \setminus I),V(H))$ is a separator in $G$. Beside this, $(M,\mathcal{R})$
contains a \emph{root node} $r$ with $R_{r} = (G,I)$ for some $I$. Moreover, for
every node $n \in V(M)$ with $R_{n} = (H,I)$ and children $n_1,\dots,n_m$ of $n$
in $M$ with $R_{n_i} = (H_i,I_i)$ for $i \in [m]$ the following properties hold:
\begin{enumerate}
\item for each $(H_i,I_i)$, we have $V(H_i) \subseteq V(H)$ and $(V(H_i) \setminus
  I_i) \subseteq (V(H) \setminus I)$ and at least one inclusion is proper,
\item for each $(H_i,I_i)$, $V(H_i)$ contains at least one vertex of $V(H)
  \setminus I$, 
\item for all distinct $(H_i,I_i)$ and $(H_j,I_j)$, $(V(H_i) \setminus I_i) \cap
  (V(H_j) \setminus I_j) = \emptyset$, and
\item each edge of $H$ is present in
  $G[I]$ or some $H_i$.
\end{enumerate}

Descriptor decompositions contain tree decompositions in the following way
\cite[Lemma III.4]{ElberfeldJT2010}. Given a descriptor decomposition
$(M,\mathcal{R})$ rooted at $r \in V(M)$, the subgraph of $M$ reachable from $r$
is a tree $T$ that can be turned into a tree decomposition $(T,\mathcal{B})$ by
setting $B_n$ for each $n \in V(T)$ to be the union of the interface $I$ of $R_n
= (H,I)$ and all vertices $x$ that are in interfaces of at least two of the $I,
I_1, \dots, I_{m}$. The \emph{width} of $(M,\mathcal{R})$ is the width of
$(T,\mathcal{B})$.

Mapping a descriptor decompositions $(M,d)$ to its tree decompositions
$(T,\mathcal{B})$ is isomorphism-invariant. Moreover, from \cite[Lemma
III.5]{ElberfeldJT2010} we know that turning descriptor decompositions into
their tree decompositions is logspace-computable. In the light of these facts,
all we need to do to, finally, prove Lemma~\ref{lem:decomp-without-clique} is to
construct an isomorphism-invariant descriptor decomposition
$(M,\mathcal{R})_{\operatorname{inv}}$ of a bounded (approximate) width $k' \in
\Nat$.

\begin{lemma}
  \label{lem:descriptor-decomposition} 
  For every $k \in \Nat$, there is a $k' \in \Nat$ and a logspace-computable and 
  isomorphism-invariant mapping that turns every graph $G$ with a distinguished
  non-edge $\{u,v\} \notin 
  E(G)$, where $G$ 
  \begin{enumerate}
  \item has tree width at most $k$,
  \item is an atom, and
  \item is improved,
  \end{enumerate}
  into a width-$k'$ descriptor decomposition $(M,\mathcal{R})$. 
\end{lemma}
\begin{proof}
  Let $\operatorname{small}, \operatorname{medium}, \operatorname{large} \colon
  \Nat \to \Nat$ be the functions satisfying
  Fact~\ref{fa:without-cliques-size-bounds}. We consider the directed graph~$M$
  whose nodes correspond to all subgraphs with interfaces $(H,I)$ where $(V(G)
  \setminus (V(H) \setminus I),V(H))$ is a separation in $G$ with separator size
  $|I| \leq \operatorname{medium}(\tw(G))$. We insert an edge from a node $n$ to a
  node~$n'$ if the graph with interface of $n'$ arises (as a component) by
  applying Fact~\ref{fa:without-cliques-size-bounds} to the one of $n$. In
  addition, we insert an edge from $n$ to a node $(G[S],S)$, which represents
  the corresponding separator $S$. We
  declare the graph with interface $(G,\emptyset)$ to be the root $r$ of
  $(M,\mathcal{R})$ and, in addition, connect it to all $(G[C \cup
  \{u,v\}],\{u,v\})$ where $C \subseteq V(G)$ is the vertex set of a component of
  $G - \{u,v\}$.

  For a constant bound on the tree width, constructing $(M,\mathcal{R})$ can be
  done by iterating over all candidate subgraphs with interfaces of a given graph
  $G$ and using Proposition~\ref{pr:separators-logspace} to construct the
  corresponding separator and connecting it with the children.

  To show that $(M,\mathcal{R})$ is a descriptor decomposition, we first observe
  that it has a root node $r$ where $R_r$ is $G$ with interface
  $\{u,v\}$. Moreover, we need to check Properties~(1) to~(4) of descriptor
  decompositions: Property (1) follows from the fact that each separator covers
  the interface $I$ and extends it. Properties~(2) and~(3) follow from the fact that
  we always consider nonempty components that are disjoint, respectively. Every
  edge is contained in $G[I]$, in $G[S]$, or in a component. The edges that are
  only in $G[S]$ are covered by $(G[S],S)$, which ensures Property (4). The bound
  on the width follows from the definition of~$S$.
\end{proof}

\begin{proof}[Proof of Lemma~\ref{lem:decomp-without-clique}] 
  To prove the lemma, we first apply Lemma~\ref{lem:descriptor-decomposition} in
  order to construct a descriptor decomposition $(M,\mathcal{R})$ that is
  isomorphism-invariant and has a bounded width. We turn it into a tree
  decomposition $(T,\mathcal{B})$ of the same width as discussed above, which is
  an isomorphism-invariant mapping. Thus, the combined mapping from $G$ with
  distinguished pair~$\{u,v\}$ to the tree decomposition is isomorphism-invariant
  as well. 
\end{proof}

\section{Isomorphism-based ordering of nested tree decompositions}  
\label{sec:nested}

Nested tree decompositions are tree decompositions whose parts are not just
bags, but where every bag is associated with a family of tree decompositions for
the bag's torso. We use polynomial-size nested tree decompositions to represent
exponential-size families of width-bounded tree decompositions that arise by
replacing bags with tree decompositions from their families. In order to solve
the isomorphism problem with the help of nested tree decompositions, we use a
recursively defined weak ordering on pairs of graphs and nested tree
decompositions. Incomparable elements in this weak ordering represent isomorphic graphs. 

In Section~\ref{sec:nested-definition} we define nested tree decompositions. To
define the ordering on nested tree decompositions in
Section~\ref{sec:ordering-nested}, we first define concepts related to weak
orderings in Section~\ref{sec:ordering-sequences}.

\subsection{Definition of nested tree decompositions}
\label{sec:nested-definition}

A \emph{nested (tree) decomposition} $\bar{D} = (T,\mathcal{B},\mathcal{D})$ for
a graph~$G$ consists of a tree decomposition $(T,\mathcal{B})$ for $G$, and a
family $\mathcal{D} = (\mathcal{D}_n)_{n \in V(T)}$ where every $\mathcal{D}_n$
is a family of tree decompositions $D \in \mathcal{D}_n$ for the torso of
$n$. Normal tree decompositions can be viewed as nested decompositions where
$\mathcal{D}_n$ is empty for every $n \in V(T)$. We adjust some terminology that
usually applies to tree decompositions for the use with nested decompositions.
Let $\bar{D} = (T,\mathcal{B},\mathcal{D})$ be a nested decomposition. The
definition of the \emph{width} of a bag $B_n$ in a nested decomposition depends
on whether $\mathcal{D}_n$ is empty or contains a set of tree decompositions. If
$|\mathcal{D}_n| = 0$, we set $\tw(B_n) := |B_n| - 1$ and $\tw(B_n) := \max\,
\{\tw(D) \mid D \in \mathcal{D}_n\}$, otherwise. The \emph{width} of $\bar{D}$
is $\tw(\bar{D}) := \max \{\tw(B_n) \mid n \in V(T)\}$. The \emph{size} of
$\bar{D}$ is $|\bar{D}| := \sum_{n \in V(T)} (1 + \max\, \{ |D| + 1 \mid D \in
\mathcal{D}_n\})$, where $|\mathcal{D}_n| = 0$ implies $\max\, \{ |D| + 1 \mid D
\in \mathcal{D}_n\} = 0$.

An \emph{(unordered) root set}~$M$ of a nested decomposition~$\bar{D} =
(T,\mathcal{B},\mathcal{D})$ is a subset~$M\subseteq B_r$ of the root bag~$B_r$
of~$\bar{D}$ with (1) $M = B_r$ in case~$|\mathcal{D}_r| = 0$, and (2) every~$D
\in \mathcal{D}_r$ has a bag $B$ with $M \subseteq B$ in case~$|\mathcal{D}_r |
>0$. An \emph{ordered root set}~$\sigma$ is an ordering of an unordered
root set.

\emph{Refining} a nested decomposition $\bar{D} = (T,\mathcal{B},\mathcal{D})$
with respect to a tree decomposition $D \in \mathcal{D}_r$ for the root $r \in
V(T)$ and an ordered root set~$\sigma$ is done as follows. First, we decompose
$G[B_r]$ using~$D$. Then, for each child bag~$B_c$ of $B_r$ in $\bar{D}$, we
find the highest bag in $D$ that contains the adhesion set $B_{\{r,c\}} = B_r
\cap B_c$ and make~$B_c$ adjacent to it. A bag of this kind exists since, by
definition, $D$ is a tree decomposition of the torso of $B_r$. We add a new bag
containing the elements of~$\sigma$. This bag is the new root of the obtained
decomposition and adjacent to the highest bag in~$D$ that contains all elements
of~$\sigma$ (in particular, this operation may change which bag of $D$ is
highest). The newly constructed nested decomposition is said to be obtained by \emph{refining}~$\bar{D}$ and
denoted by~$\bar{D}_{D,\sigma}$. The size of a nested decomposition decreases
when it is refined. That means $|\bar{D}_{D,\sigma}| < |\bar{D}|$ holds. We use this
property for proofs by induction. 

To be able to distinguish original bags and bags from refining decompositions,
we could mark the bags of $D$, which arise from the refinement step. We
circumvent the need to mark the bags by assuming that the bags~$B_n$ with empty
$\mathcal{D}_n$ are exactly the marked ones. In turn, we require from all nested
decompositions~$\bar{D}$ we consider that the set of bags~$B_n$ with empty
$\mathcal{D}_n$ form a connected subtree in~$\bar{D}$ containing the root.

\begin{proposition}
  The mapping that turns a nested decomposition $\bar{D} =
  (T,\mathcal{B},\mathcal{D})$ with decomposition $D \in \mathcal{D}_r$ and an
  ordered root set $\sigma$ into $\bar{D}_{D,\sigma}$ is logspace-computable and
  isomorphism-invariant. 
\end{proposition}

\subsection{Isomorphism-based ordering of graphs with vertex sequences}  
\label{sec:ordering-sequences}

In order to define the isomorphism-based ordering for nested decompositions, we
review notions related to composed orderings and define an ordering of graphs
with given vertex sequences.

Let $\prec$ be a \emph{weak ordering} on a set $M$, and $a \equiv a'$ denote
that two elements $a,a' \in M$ are \emph{incomparable} with respect to
$\prec$. That means, neither $a \prec a'$ nor $a' \prec a$ holds. We define the
\emph{weak ordering on sequences} from $M^* := \cup_{\,n \in \Nat}\, M^n$ with
respect to $\prec$ as follows. We set $a = a_1 \dots a_s \prec a'_1 \dots a'_t =
a'$ for $a,a' \in M^*$ if $s < t$, or $s = t$ and there is an $i \in [s]$ with
$a_i \prec a'_i$ while $a_j \equiv a_j$ holds for every $j \in [i-1]$. The
\emph{weak ordering on tuples} from $M_1 \times \dots \times M_k$ with respect
to weak orderings $\prec_i$ for sets $M_i$, respectively, is defined in the same
way except that tuples always have the same length. We denote it by
$\prec_{(1,\dots,k)}$. We  define a \emph{weak ordering on finite subsets} of $M$ by setting
$M_1 \prec M_2$ for two finite $M_1,M_2 \subseteq M$ based on comparing the
sequences we get by sorting their elements to be monotonically increasing with
respect to~$\prec$.

We write the \emph{concatenation} of sequences $\sigma$ and $\tau$
as~$\sigma\tau$. Suppose that~$(G,\sigma)$ and~$(G',\sigma')$ are pairs
consisting of graphs $G$ and $G'$ with sequences of vertices~$\sigma = v_1\dots
v_s $ and~$\sigma' = v'_1\dots v'_t $ from the respective graphs. We
set~$(G,\sigma) \seqwo (G',\sigma')$ if the sequence
$\col_G(v_1,v_1)\dots\allowbreak\col_G(v_1,v_s)\col_G(v_2,v_1)\dots\allowbreak\col_G(v_s,v_1)\dots\col_G(v_s,v_s)$
is smaller than the sequence
$\col_G(v'_1,v'_1)\dots\allowbreak\col_G(v'_1,v'_t)\allowbreak\col_G(v'_2,v'_1)\dots\allowbreak\col_G(v'_t,v'_1)\dots\col_G(v'_t,v'_t)$
with respect to the (standard) ordering $<$ of $\Nat$. We
write $(G,\sigma)\allowbreak \seqincomp (G',\sigma')$ if $(G,\sigma)$ and
$(G',\sigma')$ are \emph{incomparable} with respect to~$\seqwo$. The ordering
$\seqwo$ is logspace-computable by enumerating all pairs of vertices in
lexicographic order of the indices.

Graphs $G$ and $G'$ are \emph{isomorphic with respect to sequences of vertices}
$\sigma = v_1 \dots v_s $ and~$\sigma' = v'_1 \dots v'_t $ from the respective
graphs if $s = t$ and there is an isomorphism $\phi$ from $G$ to $G'$ with
$\phi(v_i) = v'_i$ for every $i \in [s]$. We say that $\phi$ \emph{respects}
$\sigma$ and $\sigma'$ in this case. Based on this definition, we also consider
\emph{canons of graphs with respect to vertex sequences}.

Due to the following statement, which we immediately get from the definition, we
call $\seqwo$ an \emph{isomorphism-based ordering of graphs with vertex
sequences}.

\begin{proposition}
  \label{pr:isomorphism-ordering-sequences} 
  Let $G$ and $G'$ be graphs with sequences of vertices~$\sigma = v_1 \dots v_s$
  and~$\sigma' = v'_1 \dots v'_t $ from the respective graphs.
  \begin{itemize}
  \item (``invariance''-property.) If $G$ and~$G'$ are isomorphic with respect to
    $\sigma$ and $\sigma'$, then $(G,\sigma) \seqincomp (G',\sigma')$. 
  \item (``quasi-completeness''-property.) If $(G,\sigma) \seqincomp
    (G',\sigma')$, then $G[\{v_1,\dots,v_s\}]$ and $G'[\{v'_1,\dots,v'_t\}]$ are
    isomorphic with respect to $\sigma$ and $\sigma'$.  
  \end{itemize}
\end{proposition}

\subsection{Isomorphism-based ordering of graphs with nested tree decompositions}  
\label{sec:ordering-nested}

We define an ordering of graphs with nested decompositions by recursively
ordering the child decompositions and combining this with the root bags. If a
root bag has no refining tree decompositions, this is done by trying all
possible orderings of the vertices of the bag. If the root bag has refining tree
decompositions, this is done by first refining it before going into recursion.

For each child~$c$ of the root node $r$ of a nested decomposition $D =
(T,\mathcal{B},\mathcal{D})$, we define a set~$\Pi(c)$ of orderings of a vertex
set as follows. If $|\mathcal{D}_c| = 0$, then $\Pi(c)$ contains all orderings
of the vertices of~$B_c$. If $|\mathcal{D}_c| > 0$, then $\Pi(c)$ is the set of
orderings of the adhesion set $B_{\{r,c\}} = B_r \cap B_c$. We use the sequences
from $\Pi(c)$ as ordered root sets for the child decomposition of $\bar{D}$
rooted at~$c$. 

For all tuples $(G,\bar{D},\sigma)$ and $(G',\bar{D}',\sigma')$ of graphs with
nested decompositions and ordered root sets, we define whether
$(G,\bar{D},\sigma) \decwo (G',\bar{D}',\sigma')$ holds based on the following
case distinction:

(``size''-comparison.) If~$|\bar{D}| < |\bar{D}'|$, or~$|\bar{D}| = |\bar{D}'|$
and~$|\mathcal{D}_r| < |\mathcal{D}'_{r'}|$, then set~$(G,\bar{D},\sigma) \decwo
(G',\bar{D}',\sigma')$.

(``bag''-comparison.) If~$|\bar{D}| = |\bar{D}'| = 1$ (which implies
$|\mathcal{D}_r| = |\mathcal{D}'_{r'}| = 0$), then set~$(G,\bar{D},\sigma)\decwo
(G',\bar{D}',\sigma')$ if $(G,\sigma) \seqwo (G',\sigma')$.

(``recursive''-comparison.) If~$|\bar{D}| = |\bar{D}'| > 1 $, and $|\mathcal{D}_r| =
|\mathcal{D}'_{r'}| = 0$, we compare the decompositions recursively. Let
$c_1,\ldots,c_{s}$ be the children of $r$ in $\bar{D}$ with respective child
decompositions $\bar{D}_1,\ldots,\bar{D}_s$ and subgraphs $G_1,\ldots,G_s$. Let
$c'_1,\ldots,c'_t$ be the children of $r'$ in $\bar{D}'$ with respective child
decompositions $\bar{D}'_1,\ldots,\bar{D}'_t$ and subgraphs
$G'_1,\ldots,G'_t$. Set~$(G,\bar{D},\sigma) \decwo (G',\bar{D}',\sigma')$ if the
following relation holds, which compares sets of sets that contain tuples to
which $\seqdecwo$ applies directly: 
\begin{align*}
  & \bigl\{\{((G_i,\bar{D}_i,\tau),(G,\sigma\tau)) \mid  \tau \in
    \Pi(c_i)\}\bigm| i\in [s]\bigr\}\\&
  \qquad \qquad \seqdecwo \bigl\{\{((G'_i,\bar{D}'_i,\tau'),(G',\sigma'\tau')) \mid \tau' \in
    \Pi(c'_i)\} \bigm| i\in [t] \bigr\}\, . 
\end{align*}

(``refinement''-comparison.) If~$|\bar{D}| = |\bar{D}'| > 1$, and
$|\mathcal{D}_r| = |\mathcal{D}'_{r'}| > 0$, then set $(G,\bar{D},\sigma) \decwo
(G',\bar{D}',\sigma')$ if $\{(G,\bar{D}_{D,\sigma},\sigma) \mid D \in
\mathcal{D}_r)\} \decwo \{(G',\bar{D}'_{D',\sigma'},\sigma') \mid D' \in
\mathcal{D}'_{r'})\}$ holds.

Graphs $G$ and $G'$ are \emph{isomorphic with respect to nested decompositions}
$\bar{D} = (T,\mathcal{B},\mathcal{D})$ and~$\bar{D}' =
(T',\mathcal{B'},\mathcal{D}')$ as well as ordered root sets~$\sigma$
and~$\sigma'$, respectively, if there exists an isomorphism $\phi$ from~$G$
to~$G'$ that 
\begin{enumerate}
\item respects the (normal) tree decompositions $(T,\mathcal{B})$ and
$(T',\mathcal{B}')$, 
\item respects the sequences $\sigma$ and $\sigma'$, and 
\item
for every $n \in V(T)$ there is a bijection~$\pi_n$ from $\mathcal{D}_n$ to
$\mathcal{D}_{n'}$, such that~$\phi$ restricted to $B_n$ respects~$D$
and~$\pi(D)$ for all~$D \in \mathcal{D}_n$. 
\end{enumerate} 
Based on how this definition refines
the isomorphism equivalence relation among graphs, we consider \emph{canons of
graphs with respect to nested decompositions}.

We call $\decwo$ an \emph{isomorphism-based ordering of graphs with nested
decompositions}, which is justified by the following lemma.

\begin{lemma}
  \label{lem:refining:ordering:is-equivalent:to:iso}
  Let~$(G,\bar{D},\sigma)$ and $(G',\bar{D}',\sigma')$ be tuples consisting of
  graphs with respective nested decompositions and ordered root sets.
  \begin{itemize}
  \item (``invariance''-property.) If  $G$ and $G'$ are isomorphic with respect to
    $\bar{D}$ and $\bar{D}'$ as well as $\sigma$ and $\sigma'$, then
    $(G,\bar{D},\sigma) \decincomp
    (G',\bar{D}',\sigma')$. \label{part:iso:to:equivalent}  
  \item (``quasi-completeness''-property.) If $(G,\bar{D},\sigma) \decincomp
    (G',\bar{D}',\sigma')$, then $G$ and $G'$ are isomorphic with respect to
    $\sigma$ and $\sigma'$.\label{part:equivalent:to:iso} 
  \end{itemize}
\end{lemma}
\begin{proof}
  We prove each property by induction on the sizes of $\bar{D} =
  (T,\mathcal{B})$ and $\bar{D'} = (T',\mathcal{B}')$. 

  (Proof of the ``invariance''-property.) Let $\phi$ be an isomorphism from $G$
  to $G'$ respecting $\bar{D}$~and~$\bar{D}'$ as well as
  $\sigma$~and~$\sigma'$. From the above definition we know that $\phi$ respects
  the normal tree decompositions $(T,\mathcal{B})$ and $(T',\mathcal{B}')$ for
  $G$ and $G'$, respectively, via some isomorphism $\psi \colon V(T) \to
  V(T')$.

  First of all, this implies~$|\bar{D}| = |\bar{D}'|$ as well
  as~$|\mathcal{D}_r| = |\mathcal{D}'_{r'}|$ and, hence, the ``size''-comparison does
  not distinguish $(G,\bar{D},\sigma)$ and $(G',\bar{D}',\sigma')$. If $|\bar{D}|
  = |\bar{D}'| = 1$, then we deal with the ``bag''-comparison. Since $\phi$ respects
  $\sigma$ and $\sigma'$, Proposition~\ref{pr:isomorphism-ordering-sequences}
  implies $(G,\sigma) \seqincomp (G',\sigma')$. In turn, this implies
  $(G,\bar{D},\sigma) \decincomp (G',\bar{D}',\sigma')$.

  If $|\bar{D}| = |\bar{D}'| > 1$ and $|\mathcal{D}_r| = |\mathcal{D}'_{r'}| =
  0$, we are dealing with the ``decomposition''-case. Let $c_i$ be a child of $r$
  and consider the (isomorphic) child $c'_j = \psi(c_i)$ of $r'$. Moreover,
  consider an ordering $\tau \in \Pi(c_i)$ and the (isomorphic) ordering $\tau' =
  \phi(\tau) \in \Pi(c'_j)$. By construction of $\tau$ and $\tau'$ we know
  $(G,\sigma\tau) \seqincomp (G',\sigma'\tau')$ and by applying the induction
  hypothesis we also know $(G_i,\bar{D}_i,\tau) \decincomp
  (G'_j,\bar{D}'_j,\tau')$. Since this observation holds for all children~$c_i$
  of~$r$ and all $\tau \in \Pi(c_i)$, we have~$(G,\bar{D},\sigma) \decincomp
  (G',\bar{D}',\sigma')$.

  If $|\bar{D}| = |\bar{D}'| > 1$ and $|\mathcal{D}_r| = |\mathcal{D}'_{r'}| >
  0$, we deal with the ``refinement''-case. We know that there exists a bijection
  $\pi = \pi_r$ from $\mathcal{D}_r$ to~$\mathcal{D}'_{r'}$, such that $\phi$
  restricted to the vertices from $B_r$ and $B'_{r'}$ respects every pair of tree
  decompositions $D \in \mathcal{D}_r$ and $\pi(D) \in \mathcal{D}_{r'}$ via some
  isomorphism $\psi_D$. We claim that $G$ and $G'$ are also isomorphic with
  respect to each pair of refined nested decompositions $\bar{D}_{D,\sigma}$ and
  $\bar{D}_{\pi(D),\sigma'}$ as well as ordered root sets $\sigma$ and
  $\sigma'$. Since the size of nested decompositions decreases when refining them,
  this claim implies $(G,\bar{D},\sigma) \decincomp (G',\bar{D}',\sigma')$ by
  induction. To prove it, we start to use the above isomorphism $\phi$ from $G$
  to~$G'$. We construct an isomorphism $\rho$ from the tree $T_{D,\sigma}$
  underlying $\bar{D}_{D,\sigma}$ to the tree $T'_{\pi(D),\sigma'}$ underlying
  $\bar{D}'_{\pi(D),\sigma'}$ as follows. The newly established root node of
  $T_{D,\sigma}$, whose bag is consists of the vertices of $\sigma$, is mapped to
  the newly established root node of $T'_{\pi(D),\sigma'}$, whose bag is consist
  of the vertices of $\sigma'$. Every other node is mapped according to either
  $\psi$ or $\psi_D$ depending on whether it is a tree node of either $\bar{D}$ or
  $D$, respectively. The only property we need to show is that $\rho$ preserves
  the newly established edges, which lie between nodes from $D$ and $\bar{D}$ as
  well as nodes from $\pi(D)$ and $\bar{D}'$. Let $m$ be a node of $D$ that gets
  connected to a node $c$ of $\bar{D}$ during the refinement process. Then~$m$ is
  the highest bag that contains all vertices of $B_r \cap B_c$. Due to the
  definition of $\psi$ and $\psi_D$, we know that $\psi_D(m)$ is a bag in $\pi(D)$
  that contains all vertices of $B_{\psi(r)} \cap B_{\psi(c)}$. Thus, there is
  also an edge from $\psi_D(m)$ to $\psi(c)$ in $\bar{D}'_{\pi(D),\sigma'}$. This
  proves the claim.

  (Proof of the ``quasi-completeness''-property.) For this direction, assume
  $(G,\bar{D},\sigma) \decincomp (G',\bar{D}',\sigma')$ hold. This
  implies~$|\bar{D}| = |\bar{D}'|$ and~$|\mathcal{D}_r| =
  |\mathcal{D}'_{r'}|$. If, in addition, we have $|\bar{D}| = |\bar{D}'| = 1 $,
  the statement follows from
  Proposition~\ref{pr:isomorphism-ordering-sequences}. 

  If $|\bar{D}| = |\bar{D}'| > 1$ and $|\mathcal{D}_r| = |\mathcal{D}'_{r'}| =
  0$, we are in the ``recursive''-comparison. Thus, we can choose a bijection~$\pi$
  from~$[s]$ to~$[t]$ satisfying for each $i \in [s]$
  \begin{align*}
    \{((G_i,\bar{D}_i,\tau),(G,\sigma\tau)) \mid \tau \in \Pi(c_i)\} \seqdecincomp
    \{((G_{\pi(i)},\bar{D}_{\pi(i)},\tau'),(G',\sigma\tau')) \mid \tau' \in
    \Pi(c_{\pi(i)})\} \, .
  \end{align*}
  By induction, we can choose for each~$i \in [s]$ orderings $\tau_i\in
  \Pi(c_i)$ and~$\tau'_{\pi(i)} \in \Pi(c_{\pi(i)})$, such that there is an
  isomorphism from the graph $G_i$ decomposed by~$\bar{D}_i$ to the graph
  $G'_{\pi(i)}$ decomposed by~$\bar{D}'_{\pi(i)}$ that respects $\tau_i$
  and~$\tau'_{\pi(i)}$. If subgraphs~$G_i$ and~$G_j$ of $G$ that correspond to
  child decompositions~$\bar{D}_i$ and~$\bar{D}_j$, respectively, contain a common
  vertex, then this vertex appears in~$\sigma$. Moreover, the same property
  holds for the same kind of subgraphs of $G$ and
  $\sigma'$. Since~$(G,\sigma\tau_i) \seqincomp (G',\sigma'\tau'_{\pi(i)})$, the
  mapping of common vertices agrees with the isomorphisms chosen for~$G_i$
  and~$G_j$. Thus, we can find a common extension to map~$G$ to~$G'$ that
  respects $\sigma$~and~$\sigma'$.

  If $|\bar{D}| = |\bar{D}'| > 1$ and $|\mathcal{D}_r| = |\mathcal{D}'_{r'}| >
  0$, we are in the ``refinement''-comparison and know that
  $\{(G,\bar{D}_{D,\sigma},\sigma) \mid D \in \mathcal{D}_r)\} \decincomp
  \{(G',\bar{D}'_{D',\sigma'},\sigma') \mid D' \in \mathcal{D}'_{r'})\}$
  holds. Since the size of nested decompositions decreases when refining them, we
  know by induction that $G$ and $G'$ are isomorphic with respect to $\sigma$ and $\sigma'$.
\end{proof}

\begin{remark}
  \label{rem:no:exact:correspondence} 
  The ordering~$\decwo$ is defined in order to satisfy the
  ``quasi-completeness''-property stated in
  Lemma~\ref{lem:refining:ordering:is-equivalent:to:iso}, but not a
  ``completeness''-property saying that $(G,\bar{D},\sigma) \decincomp
  (G',\bar{D}',\sigma')$ implies that $G$ and $G'$ are isomorphic with respect to
  $\sigma$ and $\sigma'$ as well as $\bar{D}$ and $\bar{D}'$, too. The reason
  behind this lies in the fact that deciding an ordering of this kind for nested
  decompositions of a bounded width is as hard as the (general) graph isomorphism
  problem. (In fact, this even holds in the case of, more restrictive, $p$-bounded
  decompositions as defined in Section~\ref{sec:canonization-ref}.) This can be
  seen by the following reduction.  Take two graphs $G$ and~$H$ for which we want
  to know whether they are isomorphic. Consider now two (empty) graphs $G'$ and
  $H'$ with $V(G') = V(G)$, $V(H') = V(H)$, and $E(G') = E(H') = \emptyset$.  For
  $G'$, construct a nested decomposition by starting with a single bag $B =
  V(G')$. Then, for each edge $\{v,w\} \in E(G)$, construct a refining
  decomposition whose tree is a star graph where the center bag equals $\{v,w\}$
  and the adjacent bags each contain a single vertex of $G'$. The nested
  decomposition for $H'$ is constructed in the same way. Since the refining
  decompositions exactly encode the edges of the respective graphs, $G$ and $H$
  are isomorphic exactly if~$G'$ and~$H'$ are isomorphic with respect to their
  nested decompositions. Thus, if~$\decincomp$ were defined as to be in exact
  correspondence with isomorphisms respecting nested decompositions, then deciding
  whether~$G'\decincomp H'$ would be graph isomorphism complete even on graphs
  without edges.
\end{remark}

\section{Computing the ordering for nested tree decompositions}   
\label{sec:canonization-ref}

We now investigate methods to space-efficiently evaluate the isomorphism based
ordering described in the previous section.  The nested decompositions we are
working with always have a bounded width. This makes it possible to implement
the ``recursive''-comparison of the isomorphism-based ordering
space-efficiently. If the child decompositions are small enough (more precisely,
they are smaller by a constant fraction in comparison to their parent), then it
is possible to store a constant amount of information, and in particular to
store orderings of the size-bounded root bag, before descending into recursion,
without exceeding a desired logarithmic space bound. If there is a large child
decomposition, of which there can be only one, then we can use Lindell's classic
technique of precomputing the recursive information before storing anything at
all.  However, for the ``refinement''-comparison, a space-efficient approach
turns out to be more challenging. In this case, the ordering asks us to compare
various refinements of the root bag. Cycling through these refinements as part
of a recursive approach requires too much space, even if the number of
decompositions is bounded by a polynomial in the size of the root bag. While it
is not clear how to remedy this difficult in general, the nested decompositions
we construct in the proofs of our main theorems satisfy an additional technical
condition, called $p$-boundedness below. This makes it possible to find a
trade-off between the recursive space requirement and the space required for
cycling through the refinements.

Let~$\bar{D}$ be a nested decomposition. Consider a bag~$n$
with~$|\mathcal{D}_n| > 1$. Let $c_1,\ldots,c_t$ be the children of $n$ sorted
by monotonically decreasing size of the respecting subdecompositions
$D_1,\ldots, D_t$.  If it exists, let $j \in [t]$ be maximal such that $G[A_n]$
with
\begin{align*}
  A_n := (B_n \cap B_{c_1}) \cup \dots \cup (B_n \cap B_{c_j}) 
\end{align*} 
is a clique, and $|D_j| > |D_{j+1}|$ holds or $j = t$ holds. Otherwise, set $j
:= 0$ and $A_n := \emptyset$. We call the children $c_1,\ldots,c_j$ of $n$ the
\emph{special children} and $A_n$ is the \emph{attachment clique of the special
children}. A nested decomposition $\bar{D}$ is \emph{$p$-bounded} for a
polynomial $p : \Nat \to \Nat$ if for every $n \in V(T)$ and non-special child
$c$ of $n$ we have $|\mathcal{D}_{n}| \leq p(|\bar{D}|/|\bar{D}_c|)$. For
non-special nodes we use the $p$-boundedness condition to trade the number of
candidate refining decompositions against the size of subdecompositions. This
enables an overall space-efficient recursion.

\begin{lemma}
  \label{lem:log:space:computable:if:nice}  
  For every $k \in \Nat$ and polynomial $p \colon \Nat \to \Nat$, there is a
  logspace \dtm{} that, on input of graphs~$G$ and~$G'$ along with
  respective nested decompositions~$\bar{D}$ and~$\bar{D}'$ and ordered root
  sets~$\sigma$ and~$\sigma'$ where $\bar{D}$ and~$\bar{D}'$
  \begin{enumerate}
  \item have width at most $k$, and 
  \item are $p$-bounded,
  \end{enumerate}
  decides $(G,\bar{D},\sigma) \decwo (G',\bar{D'},\sigma')$.
\end{lemma}

In Section~\ref{sec:lindell} we review a technique of Lindell~\cite{Lindell1992}
used to compute (composed) weak orderings on sets space-efficiently.
Section~\ref{sec:proof-logspace-ordering} contains the proof of
Lemma~\ref{lem:log:space:computable:if:nice}.

\subsection{Comparing sets via cross comparing elements} 
\label{sec:lindell}

We repeatedly apply a technique of Lindell~\cite{Lindell1992} to
compare two sets $A$ and $A'$ with respect to a weak ordering $\prec$ defined
for their elements. Comparing $A$ and $A'$ is performed by repeatedly comparing a single element of~$A$ with a single element of~$A'$.
Such a comparison is called a \emph{cross comparison}. To apply the technique
subsequently, we state it as an abstract fact as follows. Let $\prec$ be a weak
ordering for elements of a set $M$ and $A,A' \subseteq M$ finite subsets of
$M$. The \emph{cross comparison matrix} of $A = \{a_1,\dots,a_s\}$ and $A' =
\{a'_1,\dots,a'_t\}$ with respect to~$\prec$ is the binary matrix $C_{A,A'}
\in \{0,1\}^{s \times t}$ with $C_{A,A'}[i,j] = 1$ exactly if $a_i \prec a'_j$.

\begin{fact}
  \label{fa:compare-sets}
  Let $\prec$ be a weak ordering of elements of a set $M$. There is a logspace
  \dtm{} that, given the cross comparison matrix $C_{A,A'}$ for sets $A,A'
  \subseteq M$ (which are not part of the input), decides $A \prec A'$.
\end{fact}

To apply Fact~\ref{fa:compare-sets}, it is important to observe that inputs to
its \dtm{} consist only of the cross comparison matrix without encodings of the
sets $A$ and $A'$. Thus, the space used by the machine is in $O(\log(|A| \cdot
|A'|))$. The fact can be used to build an algorithm for comparing sets on top of
an algorithm for comparing individual elements.

\begin{proposition}
  \label{pr:compare-sets-oracle}
  Let $\prec$ be a weak ordering of elements of a set $M$ that can be decided
  by a \dtm{} in space at most $s(a,a')$ for every $a,a' \in
  M$. There is a \dtm{} that, given sets $A,A' \subseteq M$, decides $A \prec
  A'$ in space $O(\log(|A| \cdot |A'|) + \max\{ s(a,a') \mid a \in A \text{ and
  } a' \in A'\})$.
\end{proposition}

We apply the proposition in a scenario where the weak ordering $\prec$ is
partially known: A weak ordering~$\prec'$ is \emph{coarser} than a weak
ordering~$\prec$ if~$a_1\prec' a_2$ implies~$a_1\prec a_2$. In our application,
the weak ordering $\prec'$ compares nested subdecompositions based on their
sizes, which can be done during a logspace reduction before the more challenging
recursive computation starts.

\subsection{Proof of Lemma~\ref{lem:log:space:computable:if:nice}}
\label{sec:proof-logspace-ordering}

A child decomposition $\bar{D}_i$ of a nested decomposition $\bar{D}$ is
\emph{large} if $|\bar{D}_i| > |\bar{D}|/2$; the size bound implies that every
nested decomposition has at most one large child.

\begin{proof}[Proof of Lemma~\ref{lem:log:space:computable:if:nice}]
  The logspace procedure $\proccomp(\cdot,\cdot)$ we design implements the
  recursive definition of $\decwo$.  To streamline the recursion, it is more
  convenient to perform a computational task that is slightly more general than
  the one required by the lemma. On input of graphs $G$ and $G'$ with nested tree
  decompositions $\bar{D}$ and $\bar{D}'$ and unordered root sets $M$ and $M'$
  (not ordered root sets $\sigma$ and $\sigma'$ as described by the lemma), the
  output of $\proccomp((G,\bar{D},M),(G',\bar{D}',M'))$ is the cross comparison
  matrix of the sets $\{(G,\bar{D},\sigma)\mid \sigma \text{ is an ordering of }
  M\}$ and~$\{(G',\bar{D'},\sigma')\mid \sigma' \text{ is an ordering of } M'\}$
  with respect to $\decwo$.

  The recursive procedure starts at the root bags of both decompositions and
  descends into them in order to compute the recursively-defined ordering. To have
  access to the current positions in the decompositions, we do not use a stack,
  but maintain \emph{node pointers} to the current node and the previous node of
  the recusive process in each decomposition. These pointers direct us to nodes
  from the coarser, nested decompositions as well as to nodes from refining
  decompositions. They only require a logarithmic amount of space.  In order to be
  able to reconstruct vertex sequences as well as already refined parts of the
  decomposition, we store sequences of vertices in a bag relative to the bag using
  the pointers. Thus, to store an ordering of a bag of bounded size we only require
  a constant amount of space, provided we have the pointer to the bag at hand.

  To analyze the space requirement apart from these pointers, 
  we use two separate
  tapes. A \emph{decomposition} tape is used to store data related to the
  ``recursive''-comparison and a \emph{refinement} tape is used to store data
  related to the ``refinement''-comparison. We prove separately for each tape that the
  used space is bounded by $O(\log(|\bar{D}| +|\bar{D}'|))$.

  The procedure closely follows the definition of the isomorphism-based ordering
  and, in particular, distinguishes the same cases.

  (``size''-comparison.) By counting refining decompositions and nodes in decompositions,
  we can determine whether $|\bar{D}| \neq |\bar{D}'|$ or $|\mathcal{D}_r| \neq
  |\mathcal{D}'_{r'}|$ and in either case directly
  decide~$(G,\bar{D},\sigma)\decwo(G',\bar{D}',\sigma')$ for all orderings
  $\sigma$ and $\sigma'$ of $M$ and $M'$, respectively.

  (``bag''-comparison.) If~$|\bar{D}| = |\bar{D}'| = 1$ (and, thus,
  $|\mathcal{D}_r| = |\mathcal{D}'_{r'}| = 0$), we decide $(G,\bar{D},\sigma)
  \decwo (G',\bar{D}',\sigma')$ for all orderings $\sigma$ and $\sigma'$ of $M$
  and $M'$, respectively, based on the definition of $\seqwo$ in logspace.

  (``recursive''-comparison.) If $|\bar{D}| = |\bar{D}'| > 1$ and $|\mathcal{D}_r| =
  |\mathcal{D}'_{r'}| = 0$, we handle large and non-large children differently. 

  If~$\bar{D}$ has a large child decomposition $\bar{D}_i$ and~$\bar{D}'$ has a
  large child decomposition~$\bar{D}_j$, we first compute the comparison matrix of
  the sets~$\{(G_i,\bar{D}_i,\sigma) \mid \sigma \in \Pi(c_i)\}$
  and~$\{(G'_j,\bar{D}'_j,\sigma') \mid \sigma \in \Pi(c'_j)\}$. The descent into
  the large children can be performed without
    storing data on the decomposition tape by just updating the node pointers. The comparison matrix of (constant) size
  $s_{\textnormal{large}} \in O((k!)^2)$ is stored on the decomposition tape.

  After handling large children, we continue to compare
  $\bigl\{\{((G_i,\bar{D}_i,\tau),\sigma\tau)\mid \tau \in \Pi(c_i)\}\bigm| i\in
  [s] \bigr\}$ and $\bigl\{\{((G'_j,\bar{D}'_j,\tau'),\sigma'\tau')\mid \tau' \in
  \Pi(c'_j)\} \bigm| j\in [t]\big\}$ with respect to~$\seqdecwo$. In order to
  apply Proposition~\ref{pr:compare-sets-oracle} with a smaller space requirement,
  we we first define a weak order~$\seqdecwo'$ that is coarser than~$\seqdecwo$ by
  setting~$\big\{((G_i,\bar{D}_i,\tau),\sigma\tau)\mid \tau \in \Pi(c_i)\big\}
  \seqdecwo' \big\{((G'_j,\bar{D}'_j,\tau'),\sigma'\tau') \mid \tau' \in
  \Pi(c'_j)\big\}$ if~$|\bar{D}_i| < |\bar{D}'_j|$. Ordering the child
  decompositions with respect to $\prec$ can now be done by first ordering with
  respect to $\prec'$ with higher priority and, then, applying the cross
  comparison idea behind Proposition~\ref{pr:compare-sets-oracle} to pairs whose
  child decompositions are of the same size. Moreover, we do not need to recurse
  on large children. Note that in order to
  compare~$\{((G_i,\bar{D}_i,\tau),\sigma\tau)\mid \tau \in \Pi(c_i)\}$
  and~$\{((G'_j,\bar{D}'_j,\tau'),\sigma'\tau')\mid \tau' \in \Pi(c'_j)\}$ with
  respect to~$ \seqdecwo$ it suffices to know the result of the
  call~$\proccomp((G_i,\bar{D}_i,N),(G'_j,\bar{D'}_j,N'))$ where~$N$ and~$N'$
  denote the unordered sets of the ordered root sets in~$\Pi(c_i)$
  and~$\Pi(c'_j)$, respectively.
  
  We investigate the space requirement of the method. We first observe that there
  are at most $\ell \leq |\bar{D}|/m$ child decompositions of size
  $m$. Moreover~$m \leq |\bar{D}|/2$ for non-large children and, thus, $m\leq
  \min\{|\bar{D}|/2, |\bar{D}|/\ell\}$. By Fact~\ref{fa:compare-sets}, we can
  compare sets of size at most~$\ell$ using $O(\log(\ell^2))$ space, plus the
  space required for the recursion, which is at most~$ O(\min \{\log(|\bar{D}|/2)
  , \log(|\bar{D}|/\ell)\})$ by induction. The current orderings $\sigma$
  and~$\sigma'$ are stored using space $s_{\stack} \in O((k!)^2)$ on the
  decomposition tape. (Recall that sequences are stored relative to pointers to
  nodes of the decomposition as described above.) Defining
  $s_{\textnormal{dec}}(n)$ to be the space requirement for the recursive
  comparison of decompositions of size~$n$ where $s_{\textnormal{dec}}(1)$ depends
  on the constant $k$, we have
  \begin{align*}
    s_{\textnormal{dec}}(n) \leq s_{\textnormal{large}}+ \max_{2 \leq \ell
    \leq n} \{ s_{\stack} + O(\log(\ell^2)) + s_{\textnormal{dec}}(\lfloor
    n/\ell\rfloor) \} \in O(\log n)\, .
  \end{align*}

  (``refinement''-comparison.) In case $|\bar{D}| = |\bar{D}'| > 1$ and $|\mathcal{D}_r|
  = |\mathcal{D}'_{r'}| > 0$, we need to describe how our procedure handles the
  recursive refinement of the root bag. In the simplest case, if $|\mathcal{D}_r|
  = |\mathcal{D}'_{r'}| = 1$ holds, we choose the unique refinement and recurse
  without using any space on the refinement tape. If there are multiple refining
  decompositions for the roots, we first handle the special children before taking
  refining decompositions into account.

  For the special children, we compute and store all information about
  comparisons between them that might ever be required subsequently as follows. Let~$C$
  and~$C'$ be the attachment cliques of the special child decompositions
  of~$\bar{D}$ and~$\bar{D}'$, respectively. Consider arbitrary $P \subseteq
  \operatorname{pow}(C)$ and $P' \subseteq \operatorname{pow}(C')$, where
  $\operatorname{pow}(\cdot)$ is the power set operator. Let~$\tau$ be an ordering
  of~$C$ and~$\tau'$ be an ordering of~$C'$. Given~$\tau$ and~$\tau'$, we want to
  compare the special children whose adhesion set is in~$P$ or~$P'$.  Let~$Z_{P} =
  (G_{P},\bar{D}_{P},C)$ be defined as follows. $G_{P}$ is the subgraph of~$G$
  induced by the vertices in~$C$ and all vertices contained in a special child
  of~$\bar{D}$ with an adhesion set in~$P$ and $\bar{D}_{P}$ is the decomposition
  of~$G_{P}$ obtained by replacing the root of $\bar{D}$ with the set~$C$ and out of
  the child decompositions of the root only maintaining the special child
  decompositions with adhesion sets in~$P$. Note that~$Z_{P}$ is a graph with a
  nested decomposition and unordered root set. We define~$Z'_{\mathcal{M'}}$
  similarly for~$G'$. With this definition we compute and store all cross
  comparison matrices~$\proccomp(Z_{P},Z'_{P'})$ for all choices of~$P \subseteq
  \operatorname{pow}(C)$ and~$P' \subseteq \operatorname{pow}(C')$.

  Since the sets have only bounded size, the entire outcome of the computation
  can be stored using (constant) space~$s_{\textnormal{clique}} \in O(2^{2^k}
  k!)$. This outcome will be stored on the refinement tape. We argue that all of
  this information can be computed recursively without exceeding the logarithmic
  space bound. Indeed, there is at most one large child decomposition~$\bar{D}_L$
  of $\bar{D}$ with graph $G_L$ and at most one large child decomposition
  $\bar{D}'_L$ of $\bar{D}'$ with graph $G'_L$. Before using any space, we first
  compute the recursive call~$\proccomp((G_L,\bar{D}_L,M_L)
  ,(G'_L,\bar{D}'_L,M'_L))$, where~$M_L$ and~$M'_L$ are the unordered root sets
  of~$\bar{D}_L$ and~$\bar{D}'_L$, respectively.  The result is stored using a
  constant amount of space on the refinement tape. We then compute for all choices
  of~$P \subseteq \operatorname{pow}(C)$ and~$P' \subseteq \operatorname{pow}(C')$
  the result of~$\proccomp(Z_{P},Z'_{P'})$. Recall that~$s_{\textnormal{clique}}$
  denotes the amount of refinement space required to store the entire
  outcome. Since the size of every non-large child decompositions of~$Z_{P}$ is at
  most~$|\bar{D}|/2$, we obtain a recursion for the space satisfying
  \begin{align*}
    s_{\textnormal{refine}}(n) \leq s_{\textnormal{clique}} +
    s_{\textnormal{refine}}(\lfloor n/2 \rfloor) \in O(\log n)\, .
  \end{align*}
  
  Having computed~$\proccomp(Z_{P},Z'_{P'})$ for all choices of~$P \subseteq
  \operatorname{pow}(C)$ and~$P' \subseteq \operatorname{pow}(C')$ our goal is now
  to compare the sets~$A_\sigma = \{(G,\bar{D}_{D,\sigma},\sigma) \mid D \in
  \mathcal{D}_r)\}$ and~$A'_{\sigma'} = \{(G',\bar{D}'_{D',\sigma'},\sigma') \mid
  D' \in \mathcal{D}'_{r'})\}$ for all orderings~$\sigma$ and $\sigma'$ of $M$ and
  $M'$, respectively. Using Fact~\ref{fa:compare-sets} we can compare these sets
  using $O(\log(|\mathcal{D}_r| \cdot |\mathcal{D}'_{r'}|))$ space, which we will
  write on the refinement tape, in addition to the recursive space required to
  compare two elements~$(G,\bar{D}_{D,\sigma},\sigma)$
  and~$(G,\bar{D}_{D',\sigma'},\sigma')$.  For two such elements, whenever we would
  go into recursion, if a recursive result is contained in the precomputed
  information for special children, we will not go into recursion and rather use
  the precomputed information. Note that this means that we will never have to recursively
  descend into a special child of~$\bar{D}$ or~$\bar{D}'$ again. This observation
  uses the fact that in~$\bar{D}$ the set of bags~$n$ with~$|\mathcal{D}_n| > 0$
  (and similarly in~$\bar{D}'$) forms a connected subtree containing the root.
  
  For the refinement space consumption we now note the following. Since
  comparisons of special children are precomputed, due to the~$p$-boundedness for
  every subsequent recursive call, the size of the decompositions of the recursive
  call is at most~$\max\{|\bar{D}|/ |\mathcal{D}_r|, |\bar{D'}|/
  |\mathcal{D'}_{r'}|\}$. Thus, we obtain the recursive space requirement of
  \begin{align*}
    s_{\textnormal{refine}}(n) \leq s_{\textnormal{clique}} + \max_{2\leq \ell \leq
    n} \{ O(\log(\ell^2) + s_{\textnormal{refine}}(\lfloor n/\ell\rfloor) \} \in O(\log n)\, . 
  \end{align*}
  Thus, the total space requirement of our procedure is logarithmic.
\end{proof}

\section{Testing isomorphism and canonizing bounded tree width graphs}
\label{sec:main-theorems}

We show how to compute isomorphism-invariant width-bounded and $p$-bounded
nested decompositions for graphs of bounded tree width and, then, apply this to
prove Theorems~\ref{th:isomorphism-tw} and~\ref{th:canonization-tw}.

To compute nested decompositions, we combine the decomposition into atoms
described in Section~\ref{sec:decomposition-with-cliques} with the decomposition
of atoms into width-bounded tree decompositions described in
Section~\ref{sec:decomposition-without-cliques}.

\begin{lemma}
  \label{lem:invariant-nested-decomposition}
  For every $k \in \Nat$, there is a $k' \in \Nat$, a polynomial $p \colon \Nat
  \to \Nat$, and a logspace-computable and isomorphism-invariant mapping that
  turns every graph~$G$ of tree width at most $k$ into a nested decomposition
  $\bar{D}$ for $G$ that
  \begin{enumerate}
  \item has width at most~$k'$, and
  \item is $p$-bounded.
  \end{enumerate}
\end{lemma}
\begin{proof}
  Instead of the original input graph $G$, we work with its improved
  version, which we can compute in logspace since the tree width of $G$ is 
  bounded by~$k$. Mapping the input graph to its improved version is
  isomorphism-invariant and the improved version has exactly the same tree
  decompositions. In the following, we denote the improved version of
  the input graph by~$G$.

  Let $D = (T,\mathcal{B})$ be the isomorphism-invariant tree decomposition we
  get from~$G$ by applying Lemma~\ref{lem:main-decomposition-cliques}. Since the
  lemma guarantees that in~$D$ the adhesion sets are cliques, the torso of each
  bag is equal to the bag itself. To turn~$D$ into a nested decomposition if
  thus suffices to find a family of tree decompositions of width at most~$k$ for
  each bag.  We will apply Lemma~\ref{lem:decomp-without-clique} to find such a
  family.  Since $D$ decomposes an improved graph and the adhesion sets are
  cliques, every $G[B_n]$ for $n \in V(T)$ is also improved.

  Thus, based on $D$, we construct a nested decomposition $\bar{D}$ by
  considering every node $n$ of $D$ and defining an isomorphism-invariant family
  $\mathcal{D}_n$ of tree decompositions of the bag $B_n$. If~$B_n$ has size at
  most~$k+1$, we let the family $\mathcal{D}_n$ consist of a single tree
  decomposition that is just $B_n$. Note that by this choice, the bag~$B_n$
  satisfies both the width bounded and the~$p$-boundedness restriction (for every
  polynomial $p$ with~$p(i) \geq 1$ for all~$i\in\Nat$). If the size of $B_n$
  exceeds $k+1$, we would like to apply Lemma~\ref{lem:decomp-without-clique} to
  further decompose~$B_n$. However, for the lemma, we need a pair~$\{u,v\}\notin
  E(G)$ in~$B_n$ to serve as the root of the decomposition.  We cannot simply
  iterate over all~$\{u,v\}\notin E(G)$ in~$B_n$ since the result may violate
  the~$p$-boundedness condition. We proceed as follows. Let $c_1,\ldots, c_t$ be the
  children of $n$ sorted by decreasing size of the respecting child decompositions
  $D_{c_1},\ldots, D_{c_t}$. If it exists, let $j \in [t]$ be the maximum, such that
  $G[A_n]$ with
  \begin{align*}
    A_n := (B_n \cap B_{c_1}) \cup \dots \cup (B_n \cap B_{c_j})
  \end{align*}
  is a clique, and $|D_{c_j}| > |D_{c_{j+1}}|$ holds or
  $j = t$ holds.  Otherwise, set $j := 0$ and $A_n := \emptyset$.
  Thus,~$A_n$ is the attachment clique of the special children as defined above.
  We construct a
  collection of tree decompositions $\mathcal{D}_n$ for $B_n$ based on
  whether we have $j < t$ or $j = t$. If $j < t$,
  let $m \geq 1$ be the largest integer with $|D_{c_{j+1}}| =
  |D_{c_{j+m}}|$. By construction, we can find at least one and at most
  $((k+1)(m+1))^2$ pairs of nonadjacent vertices~$\{u,v\}$ in~$G[A_n']$ for 
  \begin{align*}
    A_n' := A_n \cup (B_n \cap B_{c_{j+1}}) \cup \dots \cup (B_n \cap
    B_{c_{j+m}})\, .
  \end{align*}
  We define $\mathcal{D}_n$ to be the collection of tree decompositions we
  obtained by applying Lemma~\ref{lem:decomp-without-clique} to $G[B_n]$ with
  pairs $\{u,v\}$ of nonadjacent vertices in~$G[A_n']$. We have $|\mathcal{D}_n|
  \leq ((k+1)(m+1))^2$. This set of decompositions satisfies the~$p$-boundedness
  restriction with the polynomial $p(m) = ((k+1)(m+1))^2$. If $j = t$, we consider
  every pair of nonadjacent vertices~$\{u,v\}$ in~$B_n$. Again, for every such
  $\{u,v\}$, we construct a decomposition for $G[B]$ using
  Lemma~\ref{lem:decomp-without-clique}. We have $1 \leq |\mathcal{D}_n| \leq
  |B_n|^2$ in this case, satisfying the~$p$-boundedness condition, since~$B_n$
  only has special children.  Since the construction of the collections
  $\mathcal{D}_n$ is isomorphism-invariant, the entire construction is
  isomorphism-invariant.
\end{proof}

We have assembled all the required tools to prove our main theorems showing that
isomorphism of graphs of bounded tree width and canonization of graphs of
bounded tree width can be performed in logarithmic space.

\begin{proof}[Proof of Theorem~\ref{th:isomorphism-tw}] 
  Given two graphs~$G$ and~$G'$, by
  Lemma~\ref{lem:invariant-nested-decomposition} we can compute in logarithmic space
  isomorphism-invariant $p$-bounded nested decompositions~$\bar{D}$
  and~$\bar{D}'$. By Lemma~\ref{lem:refining:ordering:is-equivalent:to:iso}, the
  graphs are isomorphic if and only if there exist ordered root sets~$\sigma$
  and~$\sigma'$ with~$(G,\bar{D},\sigma) \decincomp
  (G',\bar{D}',\sigma')$. By Lemma~\ref{lem:log:space:computable:if:nice}, this
  can be checked in logarithmic space by iterating over all suitable choices
  of~$\sigma$ and~$\sigma'$.

  The $\Class{L}$-hardness for every positive $k \in \Nat$ follows from the
  $\Class{L}$-hardness of the isomorphism problem for trees (connected graphs of
  tree width at most 1) proved by Jenner et al.~\cite{Jenneretal2003}.
\end{proof}

For our canonization procedure we would like to recursively order the vertices
according to $\decwo$. However, due to the fact that there is no exact
correspondence between $\decwo$ and isomorphism for graphs with nested
decompositions and ordered root sets (recall Remark~\ref{rem:no:exact:correspondence} and that we only have a 
``quasi-completeness''-property not a ``completeness''-property), 
we need to ensure that the process is canonical. However, as the following proof shows, to ensure canonicity it is sufficient to work with an isomorphism-invariant decomposition.

\begin{proof}[Proof of Theorem~\ref{th:canonization-tw}]  
  We use the isomorphism-invariant mapping from
  Lemma~\ref{lem:invariant-nested-decomposition} to turn $G$ into a width-bounded
  and $p$-bounded nested decomposition $\bar{D} =
  (T,\mathcal{B},\mathcal{D})$. The canonical sequence of $G$'s vertices is based
  on $(G,\bar{D},\sigma)$ where $\sigma$ is the empty vertex sequence. In order to
  compute a canonical sequence with respect to $\decwo$, we repeatedly apply
  Lemma~\ref{lem:log:space:computable:if:nice}.

  If $|\mathcal{D}_r| = 0$, let~$\bar{D}_1,\dots,\bar{D}_s$ be the child
  decompositions of~$G$ containing at least one vertex that is not in~$\sigma$. We
  obtain an order on them by defining~$\bar{D}_i < \bar{D}_j$ if
  \begin{align*}
    \{((G_i,\bar{D}_i,\tau),(G,\sigma\tau)) \mid \tau \in \Pi(c_i)\} \seqdecwo
    \{((G_j,\bar{D}_j,\tau),(G,\sigma\tau)) \mid \tau \in \Pi(c_j)\} \, .
  \end{align*}
  Ties are broken arbitrarily, for example by considering the smallest vertex in the child according to the
  input ordering. For each child~$\bar{D}_i$ we compute an ordering~$\tau_i\in \Pi(c_i)$
  that minimizes~$(G_i,\bar{D}_i,\tau_i)$. We recursively create a canonical
  sequence outputting the canonical sequence of~$(G_i,\bar{D}_i,\tau_i)$ for each
  child in the order of children just defined.
  
  If $|\mathcal{D}_r| > 0$, we iterate over all decompositions
  in~$\mathcal{D}_r$ and choose a tuple from~$\{(G,\bar{D}_{D,\sigma},\sigma) \mid
  D \in \mathcal{D}_{B})\}$ that is minimal with respect to~$\decwo$. Ties are,
  again, broken based on the input ordering. For computing the canonical sequence
  we continue recursively on a minimal $(G,\bar{D}_{D,\sigma},\sigma)$ only. In
  order to obtain a canonical sequence, we alter the nested decomposition slightly
  whenever we go into the recursion using colored edges. More specifically,
  Lemma~\ref{lem:invariant-nested-decomposition} constructs $D$ based on two
  vertices~$u$ and $v$ that form a distinguished non-edge. We insert an edge
  between $u$ and $v$ and color it with a color that does not appear in $G$ (for
  example, we use~$-2$). In other words, we set~$\col_G(u,v) = -2$. This
  modification is isomorphism-invariant based on the choice of $D$. The new edge
  is covered by a bag of $D$ by construction. Inserting the edge only depends on
  $D$ and, thus, it is stored recursively in an implicit way. The modification has
  the consequence that distinguished edges are preserved under isomorphism.

  The logspace-computability of the sequence follows from
  Lemma~\ref{lem:log:space:computable:if:nice}. Thus, we are left to prove that
  the sequence is canonical. For this we need to show that whenever a tie is
  broken arbitrarily between two options, then the two options are
  equivalent. There are two situations when a tie can occur.
  
  For the first one suppose~$\{(G_i,\bar{D}_i,\tau),(G,\sigma\tau)\mid \tau \in
  \Pi(c_i)\} \seqdecincomp \{((G_j,\bar{D}_j,\tau),(G,\sigma\tau)) \mid \tau' \in
  \Pi(c_j)\}$ for two child decompositions both containing a vertex not
  in~$\sigma$.  By~Lemma~\ref{lem:refining:ordering:is-equivalent:to:iso}, there
  is an isomorphism from the graph induced by the vertices in~$\bar{D}_i$ to the
  graph induced by the vertices in~$\bar{D}_j$ fixing~$\sigma$. This extends to an
  automorphism of~$G$ by fixing all vertices neither in~$\bar{D}_i$
  nor~$\bar{D}_j$.  Since~$\bar{D}$ is isomorphism-invariant this automorphism
  respects~$\bar{D}$ therefore mapping~$\bar{D}_i$ to~$\bar{D}_j$.

  For the other case where a tie can occur,
  suppose~$(G_i,(\bar{D}_i)_{D,\sigma},\sigma) \seqdecincomp
  ((G_j),(\bar{D}_j)_{D',\sigma},\sigma)$. By
  Lemma~\ref{lem:refining:ordering:is-equivalent:to:iso}, there is an isomorphism
  from $G_i$ to $G_j$. This isomorphism preserves the distinguished edge. This
  isomorphism extends to an automorphism of~$G$ that fixes all vertices that
  neither appear in~$(\bar{D}_i)_{D,\sigma}$ nor
  in~$(\bar{D}_j)_{D',\sigma}$. Since~$\bar{D}$ is isomorphism-invariant, this
  automorphism of~$G$ respects~$\bar{D}$ and since the distinguished edge is
  preserved it maps~$(\bar{D}_i)_{D,\sigma}$ to~$(\bar{D}_j)_{{D'},\sigma}$.

  This shows that the computed sequence is canonical.
\end{proof}

\section{Conclusion}
\label{sec:conclusion}

\paragraph{Summary.} 

We showed how to canonize and compute canonical labelings for graphs of bounded
tree width in logspace, and this implies that deciding isomorphic graphs and
computing isomorphisms can be done in logspace for graphs of bounded tree
width. For the proof we first developed a tree decomposition into
clique-separator-free subgraphs that is isomorphism-invariant and
logspace-computable. Then we showed how to compute, for each bag, an
isomorphism-invariant family of width-bounded tree decompositions in
logspace. Finally, we combined both decomposition approaches to construct nested
tree decompositions and developed a recursive canonization procedure that works
on nested tree decompositions.

\paragraph{Outlook.}

Testing $\Lang{isomorphism}$ for graphs that are embeddable into the
plane~\cite{Dattaetal2009} as well as any fixed surface~\cite{ElberfeldK2014}
can be done in logspace. These graph classes can be described in terms of
forbidding fixed minors, which also holds for classes of graphs whose tree width
is bounded by constants. This opens up the question of whether these logspace
results generalize to any class of graphs excluding fixed minors. For these
classes polynomial-time $\Lang{isomorphism}$ procedures are
known~\cite{Ponomarenko1991}. Partial results are known for graphs that exclude
the minors $K_5$ or $K_{3,3}$~\cite{DattaNTW2009}, but the available techniques
are tailored to the respective graph classes. Looking at algorithmic proofs
related to the structure of graphs excluding fixed
minors~\cite{Grohe2012,GroheKR2013}, it seems promising to
combine the earlier logspace approach for embeddable
graphs~\cite{ElberfeldK2014} with our logspace approach for bounded tree width
graphs.

Of course, the basic question concerning the complexity of $\Lang{isomorphism}$
on general graphs remains open. With respect to this question, our work might
help to clarify the difference between graphs to which the best known
complexity-theoretic lower bounds for \Lang{isomorphism}~\cite{Toran2004} apply,
which are given in terms of classes defined via nondeterministic
logarithmic-space-bounded Turing machines, and graphs for which (deterministic)
logspace algorithms are possible.

\addcontentsline{toc}{section}{References}
\bibliographystyle{abbrvurl}
\bibliography{main}

\begin{thebibliography}{10}

\bibitem{Arvindetal2012}
V.~Arvind, B.~Das, J.~K{\"{o}}bler, and S.~Kuhnert.
\newblock The isomorphism problem for $k$-trees is complete for logspace.
\newblock {\em Information and Computation}, 217:1--11, 2012.
\newblock \href {http://dx.doi.org/10.1016/j.ic.2012.04.002}
  {\path{doi:10.1016/j.ic.2012.04.002}}.

\bibitem{Arvindetal2008}
V.~Arvind, B.~Das, and J.~Köbler.
\newblock A logspace algorithm for partial 2-tree canonization.
\newblock In {\em Proceedings of the 3rd International Computer Science
  Symposium in Russia (CSR 2008)}, number 5010 in Lecture Notes in Computer
  Science, pages 40--51. Springer, 2008.
\newblock \href {http://dx.doi.org/10.1007/978-3-540-79709-8_8}
  {\path{doi:10.1007/978-3-540-79709-8_8}}.

\bibitem{Bodlaender1990}
H.~L. Bodlaender.
\newblock Polynomial algorithms for graph isomorphism and chromatic index on
  partial $k$-trees.
\newblock {\em J. Algorithms}, 11(4):631--643, 1990.
\newblock \href {http://dx.doi.org/10.1016/0196-6774(90)90013-5}
  {\path{doi:10.1016/0196-6774(90)90013-5}}.

\bibitem{BoppanaHZ1987}
R.~B. Boppana, J.~Hastad, and S.~Zachos.
\newblock Does {$\Class{co\textnormal{-}NP}$} have short interactive proofs?
\newblock {\em Inf. Process. Lett.}, 25(2):127--132, May 1987.
\newblock \href {http://dx.doi.org/10.1016/0020-0190(87)90232-8}
  {\path{doi:10.1016/0020-0190(87)90232-8}}.

\bibitem{DasER2015}
B.~Das, M.~Enduri, and I.~Reddy.
\newblock Logspace and {$\Class{FPT}$} algorithms for graph isomorphism for
  subclasses of bounded tree-width graphs.
\newblock In {\em 9th International Workshop on Algorithms and Computation
  (WALCOM 2015)}, volume 8973 of {\em Lecture Notes in Computer Science}, pages
  329--334. Springer, 2015.
\newblock \href {http://dx.doi.org/10.1007/978-3-319-15612-5_30}
  {\path{doi:10.1007/978-3-319-15612-5_30}}.

\bibitem{DasTW2012}
B.~Das, J.~Tor\'{a}n, and F.~Wagner.
\newblock Restricted space algorithms for isomorphism on bounded treewidth
  graphs.
\newblock {\em Information and Computation}, 217:71--83, 2012.
\newblock \href {http://dx.doi.org/10.1016/j.ic.2012.05.003}
  {\path{doi:10.1016/j.ic.2012.05.003}}.

\bibitem{Dattaetal2009}
S.~Datta, N.~Limaye, P.~Nimbhorkar, T.~Thierauf, and F.~Wagner.
\newblock Planar graph isomorphism is in log-space.
\newblock In {\em Proceedings of the 24th Annual IEEE Conference on
  Computational Complexity (CCC 2009)}, pages 203--214. {IEEE} Computer
  Society, 2009.
\newblock \href {http://dx.doi.org/10.1109/CCC.2009.16}
  {\path{doi:10.1109/CCC.2009.16}}.

\bibitem{DattaNTW2009}
S.~Datta, P.~Nimbhorka, T.~Thierauf, and F.~Wagner.
\newblock Graph isomorphism for {$K_{3,3}$}-free and {$K_5$}-free graphs is in
  log-space.
\newblock In {\em Proceedings of the 29th Annual IARCS Conference on
  Foundations of Software Technology and Theoretical Computer Science (FSTTCS
  2009)}, volume~4 of {\em LIPIcs}, pages 145--156. Schloss Dagstuhl -
  Leibniz-Zentrum fuer Informatik, 2009.
\newblock \href {http://dx.doi.org/10.4230/LIPIcs.FSTTCS.2009.2314}
  {\path{doi:10.4230/LIPIcs.FSTTCS.2009.2314}}.

\bibitem{Diestel2005}
R.~Diestel.
\newblock {\em Graph Theory}, volume 173 of {\em Graduate Texts in
  Mathematics}.
\newblock Springer, 3rd edition, 2005.

\bibitem{ElberfeldJT2010}
M.~Elberfeld, A.~Jakoby, and T.~Tantau.
\newblock Logspace versions of the theorems of {Bodlaender} and {Courcelle}.
\newblock In {\em Proceedings of the 51st Annual {IEEE} Symposium on
  Foundations of Computer Science (FOCS 2010)}, pages 143--152. {IEEE} Computer
  Society, 2010.
\newblock \href {http://dx.doi.org/10.1109/FOCS.2010.21}
  {\path{doi:10.1109/FOCS.2010.21}}.

\bibitem{ElberfeldK2014}
M.~Elberfeld and K.~Kawarabayashi.
\newblock Embedding and canonizing graphs of bounded genus in logspace.
\newblock In {\em Proceedings of the 46th Annual ACM Symposium on Theory of
  Computing (STOC 2014)}, pages 383--392, New York, NY, USA, 2014. ACM.
\newblock \href {http://dx.doi.org/10.1145/2591796.2591865}
  {\path{doi:10.1145/2591796.2591865}}.

\bibitem{FilottiM1980}
I.~S. Filotti and J.~N. Mayer.
\newblock A polynomial-time algorithm for determining the isomorphism of graphs
  of fixed genus.
\newblock In {\em Proceedings of the 12th Annual ACM Symposium on Theory of
  Computing (STOC 1980)}, pages 236--243, New York, NY, USA, 1980. ACM.
\newblock \href {http://dx.doi.org/10.1145/800141.804671}
  {\path{doi:10.1145/800141.804671}}.

\bibitem{Grohe2000}
M.~Grohe.
\newblock Isomorphism testing for embeddable graphs through definability.
\newblock In {\em Proceedings of the 32nd Annual ACM Symposium on Theory of
  Computing (STOC 2000)}, pages 63--72. ACM, 2000.
\newblock \href {http://dx.doi.org/10.1145/335305.335313}
  {\path{doi:10.1145/335305.335313}}.

\bibitem{Grohe2012}
M.~Grohe.
\newblock Fixed-point definability and polynomial time on graphs with excluded
  minors.
\newblock {\em J. {ACM}}, 59(5):27, 2012.
\newblock \href {http://dx.doi.org/10.1145/2371656.2371662}
  {\path{doi:10.1145/2371656.2371662}}.

\bibitem{GroheKR2013}
M.~Grohe, K.~Kawarabayashi, and B.~A. Reed.
\newblock A simple algorithm for the graph minor decomposition -- logic meets
  structural graph theory.
\newblock In {\em Proceedings of the 23rd Annual {ACM/SIAM} Symposium on
  Discrete Algorithms (SODA 2012)}, pages 414--431. SIAM, 2013.

\bibitem{GroheV2006}
M.~Grohe and O.~Verbitsky.
\newblock Testing graph isomorphism in parallel by playing a game.
\newblock In {\em Proceedings of 33rd International Colloquium on Automata,
  Languages and Programming (ICALP 2006)}, pages 3--14, 2006.
\newblock \href {http://dx.doi.org/10.1007/11786986_2}
  {\path{doi:10.1007/11786986_2}}.

\bibitem{HopcroftW1974}
J.~E. Hopcroft and J.~K. Wong.
\newblock Linear time algorithm for isomorphism of planar graphs (preliminary
  report).
\newblock In {\em Proceedings of the 6th Annual ACM Symposium on Theory of
  Computing (STOC 1974)}, pages 172--184, New York, NY, USA, 1974. ACM.
\newblock \href {http://dx.doi.org/10.1145/800119.803896}
  {\path{doi:10.1145/800119.803896}}.

\bibitem{Jenneretal2003}
B.~Jenner, J.~K\"{o}bler, P.~McKenzie, and J.~Tor\'{a}n.
\newblock Completeness results for graph isomorphism.
\newblock {\em J. Comput. Syst. Sci.}, 66(3):549--566, 2003.
\newblock \href {http://dx.doi.org/10.1016/S0022-0000(03)00042-4}
  {\path{doi:10.1016/S0022-0000(03)00042-4}}.

\bibitem{Jones1975}
N.~D. Jones.
\newblock Space-bounded reducibility among combinatorial problems.
\newblock {\em Journal of Computer and System Sciences}, 11(1):68--85, 1975.
\newblock \href {http://dx.doi.org/10.1016/S0022-0000(75)80050-X}
  {\path{doi:10.1016/S0022-0000(75)80050-X}}.

\bibitem{LadnerL1976}
R.~E. Ladner and N.~A. Lynch.
\newblock Relativization of questions about log space computability.
\newblock {\em Theory of Computing Systems}, 10:19--32, 1976.
\newblock \href {http://dx.doi.org/10.1007/BF01683260}
  {\path{doi:10.1007/BF01683260}}.

\bibitem{Leimer1993}
H.-G. Leimer.
\newblock Optimal decomposition by clique separators.
\newblock {\em Discrete Mathematics}, 113(1–3):99 -- 123, 1993.
\newblock \href {http://dx.doi.org/10.1016/0012-365X(93)90510-Z}
  {\path{doi:10.1016/0012-365X(93)90510-Z}}.

\bibitem{Lindell1992}
S.~Lindell.
\newblock A logspace algorithm for tree canonization (extended abstract).
\newblock In {\em Proceedings of the 24th Annual ACM Symposium on Theory of
  Computing (STOC 1992)}, pages 400--404, New York, NY, USA, 1992. ACM.
\newblock \href {http://dx.doi.org/10.1145/129712.129750}
  {\path{doi:10.1145/129712.129750}}.

\bibitem{Lokshtanovetal2014a}
D.~Lokshtanov, M.~Pilipczuk, M.~Pilipczuk, and S.~Saurabh.
\newblock Fixed-parameter tractable canonization and isomorphism test for
  graphs of bounded treewidth.
\newblock In {\em Proceedings of the 55th {IEEE} Symposium on Foundations of
  Computer Science (FOCS 2014)}, pages 186--195. {IEEE} Computer Society, 2014.
\newblock \href {http://dx.doi.org/10.1109/FOCS.2014.28}
  {\path{doi:10.1109/FOCS.2014.28}}.

\bibitem{Lokshtanovetal2014b}
D.~Lokshtanov, M.~Pilipczuk, M.~Pilipczuk, and S.~Saurabh.
\newblock Fixed-parameter tractable canonization and isomorphism test for
  graphs of bounded treewidth.
\newblock abs/1404.0818, 2014.
\newblock URL: \url{http://arxiv.org/abs/1404.0818}.

\bibitem{Miller1980}
G.~L. Miller.
\newblock Isomorphism testing for graphs of bounded genus.
\newblock In {\em Proceedings of the 12th Annual ACM Symposium on Theory of
  Computing (STOC 1980)}, pages 225--235, New York, NY, USA, 1980. ACM.
\newblock \href {http://dx.doi.org/10.1145/800141.804670}
  {\path{doi:10.1145/800141.804670}}.

\bibitem{OtachiS2014}
Y.~Otachi and P.~Schweitzer.
\newblock Reduction techniques for graph isomorphism in the context of width
  parameters.
\newblock In {\em Proceddings of the 14th Scandinavian Symposium and Workshops
  on Algorithm Theory (SWAT 2014)}, pages 368--379, 2014.
\newblock \href {http://dx.doi.org/10.1007/978-3-319-08404-6_32}
  {\path{doi:10.1007/978-3-319-08404-6_32}}.

\bibitem{Ponomarenko1991}
I.~N. Ponomarenko.
\newblock The isomorphism problem for classes of graphs closed under
  contraction.
\newblock {\em Journal of Mathematical Sciences}, 55(2):1621--1643, 1991.
\newblock \href {http://dx.doi.org/10.1007/BF01098279}
  {\path{doi:10.1007/BF01098279}}.

\bibitem{Reingold2008}
O.~Reingold.
\newblock Undirected connectivity in log-space.
\newblock {\em Journal of the {ACM}}, 55(4):1--24, 2008.
\newblock \href {http://dx.doi.org/10.1145/1391289.1391291}
  {\path{doi:10.1145/1391289.1391291}}.

\bibitem{Schoening1988}
U.~Schöning.
\newblock Graph isomorphism is in the low hierarchy.
\newblock {\em Journal of Computer and System Sciences}, 37(3):312--323, 1988.
\newblock \href {http://dx.doi.org/10.1016/0022-0000(88)90010-4}
  {\path{doi:10.1016/0022-0000(88)90010-4}}.

\bibitem{StockmeyerM1973}
L.~J. Stockmeyer and A.~R. Meyer.
\newblock Word problems requiring exponential time (preliminary report).
\newblock In {\em Proceedings of the 5th Annual {ACM} Symposium on Theory of
  Computing (STOC 1973)}, pages 1--9. ACM, 1973.
\newblock \href {http://dx.doi.org/10.1145/800125.804029}
  {\path{doi:10.1145/800125.804029}}.

\bibitem{Tarjan1971}
R.~E. Tarjan.
\newblock A {$V^2$} algorithm for determining isomorphism of planar graphs.
\newblock {\em Information Processing Letters}, 1(1):32--34, 1971.
\newblock \href {http://dx.doi.org/10.1016/0020-0190(71)90019-6}
  {\path{doi:10.1016/0020-0190(71)90019-6}}.

\bibitem{Toran2004}
J.~Tor\'{a}n.
\newblock On the hardness of graph isomorphism.
\newblock {\em SIAM Journal on Computing}, 33(5):1093--1108, 2004.
\newblock \href {http://dx.doi.org/10.1137/S009753970241096X}
  {\path{doi:10.1137/S009753970241096X}}.

\bibitem{Vollmer1999}
H.~Vollmer.
\newblock {\em Introduction to Circuit Complexity: A Uniform Approach}.
\newblock Springer, Berlin Heidelberg, 1999.

\bibitem{Wagner2011}
F.~Wagner.
\newblock Graphs of bounded treewidth can be canonized in {$\Class{AC}^1$}.
\newblock In {\em Proceedings of the 6th International Computer Science
  Symposium in Russia (CSR 2011)}, number 6651 in Lecture Notes in Computer
  Science, pages 209--222. Springer, 2011.
\newblock \href {http://dx.doi.org/10.1007/978-3-642-20712-9_16}
  {\path{doi:10.1007/978-3-642-20712-9_16}}.

\end{thebibliography}

\end{document}